\documentclass{amsart}

\title{\phantom{ }\\Computable Characterisations of \\Scaled Relative Graphs of Closed Operators}

\usepackage{amsmath, amssymb, amsfonts, mathtools}
\usepackage{amsthm}
\usepackage{amsrefs}
\usepackage{enumitem}
\usepackage{accents}
\usepackage{mathrsfs}
\usepackage[a4paper,margin=1.32in]{geometry}
\mathtoolsset{showonlyrefs=true}  
\usepackage{hyperref}
\hypersetup{
    colorlinks = true,
    linkcolor  = DutchOrange,    
    citecolor  = DutchOrange,    
    urlcolor   = DutchOrange,    
}

\newtheorem{theorem}{Theorem}
\newtheorem{corollary}{Corollary}
\newtheorem{remark}{Remark}
\newtheorem{proposition}{Proposition}
\newtheorem{definition}{Definition}
\newtheorem{lemma}{Lemma} 

\usepackage{tikz,pgfplots,pgfplotstable} 
\usepackage{tikz-cd}
\usepgfplotslibrary{groupplots}
\usetikzlibrary{arrows.meta}
\usetikzlibrary{patterns.meta} 
\usepackage{graphicx}  
\usepackage{subcaption}
\pgfplotsset{compat=newest}

\usepackage[ruled,vlined]{algorithm2e}

\usepackage[svgnames]{xcolor} 
\DefineNamedColor{named}{LUblue}{cmyk}{1,0.85,0.05,0.22}
\DefineNamedColor{named}{LUbronze}{cmyk}{0.09,0.57,1,0.41}
\DefineNamedColor{named}{LuPastelle}{cmyk}{0,0.15,0.05,0.0}
\DefineNamedColor{named}{LuPastelle2}{cmyk}{0.29,0.02,0.24,0.03}
\DefineNamedColor{named}{LuPastelle3}{cmyk}{0.24,0.03,0.07,0.02}
\DefineNamedColor{named}{gentledark}{gray}{.15}
\colorlet{LUdarkblue}{LUblue!60!black}
\DefineNamedColor{named}{DutchOrange}{rgb}{1,0.498,0}

\usepackage{microtype}

\usepackage[numbers]{natbib}

\usepackage[acronym]{glossaries}
\newacronym{lti}{LTI}{linear time-invariant}
\newacronym[plural=LMIs,
             longplural=linear matrix inequalities]{lmi}{LMI}{linear matrix inequality}
\newacronym[plural=SRGs,
             longplural=Scaled Relative Graphs]{srg}{SRG}{Scaled Relative Graph}

\def\BibTeX{{\rm B\kern-.05em{\sc i\kern-.025em b}\kern-.08em
    T\kern-.1667em\lower.7ex\hbox{E}\kern-.125emX}}

\usepackage{etoolbox}
\makeatletter
\patchcmd{\equation}{\@beginparpenalty}{\@lowpenalty}{}{}
\patchcmd{\endequation}{\@endparpenalty}{\@lowpenalty}{}{}
\makeatother

\makeatletter

\def\HyperAbsTol{0.5pt}
\def\HyperRelTol{1e-6}

\def\hyper@x#1,#2\relax{#1} 
\def\hyper@y#1,#2\relax{#2} 
\def\hyper@coords#1{#1}     

\pgfmathdeclarefunction{safeDiv}{2}{%
  \pgfmathparse{ifthenelse(abs(#2) < \hyper@tol, 0, #1/#2)}%
}

\def\hyper@computer#1#2{%
  \edef\hyper@toscan{(#1)}%
  \tikz@scan@one@point\hyper@coords\hyper@toscan
  \edef\hyper@sx{\the\pgf@x}%
  \edef\hyper@sy{\the\pgf@y}%
  \edef\hyper@toscan{(#2)}%
  \tikz@scan@one@point\hyper@coords\hyper@toscan
  \edef\hyper@ex{\the\pgf@x}%
  \edef\hyper@ey{\the\pgf@y}%

  \pgfmathsetmacro{\hyper@mx}{(\hyper@ex + \hyper@sx)/2}%
  \pgfmathsetmacro{\hyper@my}{(\hyper@ey + \hyper@sy)/2}%
  \pgfmathsetmacro{\hyper@dy}{\hyper@ey - \hyper@sy}%
  \pgfmathsetmacro{\hyper@dx}{\hyper@ex - \hyper@sx}%

  \pgfmathsetmacro{\hyper@L}{veclen(\hyper@dx,\hyper@dy)}%
  \pgfmathsetmacro{\hyper@tol}{max(\HyperAbsTol, \HyperRelTol*\hyper@L)}%

  \pgfmathtruncatemacro{\isvertical}{abs(\hyper@dx) < \hyper@tol}%
  \pgfmathtruncatemacro{\ishorizontal}{abs(\hyper@dy) < \hyper@tol && (abs(\hyper@dx) < 0.5)}%

  \ifnum\numexpr\isvertical+\ishorizontal>0\relax
    \edef\hyper@cmd{-- (\tikztotarget)}%
  \else
    \pgfmathsetmacro{\hyper@t}{safeDiv(\hyper@my, \hyper@dx)}%
    \pgfmathsetmacro{\hyper@cx}{\hyper@mx + \hyper@t * \hyper@dy}%
    \pgfmathsetmacro{\hyper@radius}{veclen(\hyper@cx - \hyper@sx, \hyper@sy)}%
    \pgfmathsetmacro{\hyper@sangle}{180 - atan2(\hyper@sy,\hyper@cx-\hyper@sx)}%
    \pgfmathsetmacro{\hyper@eangle}{180 - atan2(\hyper@ey,\hyper@cx-\hyper@ex)}%
    \edef\hyper@cmd{arc[radius=\hyper@radius pt, start angle=\hyper@sangle, end angle=\hyper@eangle]}%
  \fi
}

\tikzset{%
  hyperbolic disc radius/.initial={1cm},
  hyperbolic plane/.style={
    to path={
      \pgfextra{\hyper@computer\tikztostart\tikztotarget}%
      \hyper@cmd
    }
  },
  hyperbolic disc/.style={
    to path={
      \pgfextra{\hyper@disc@computer\tikztostart\tikztotarget}%
      \hyper@cmd
    }
  },
  hyperbolic plane target angle/.initial={},
}

\pgfmathdeclarefunction{bkx}{2}{%
\begingroup
 \pgfmathparse{1-(2/(1+#1*#1+#2*#2))}
 \pgfmath@smuggleone\pgfmathresult%
\endgroup}

\pgfmathdeclarefunction{bky}{2}{%
\begingroup
 \pgfmathparse{-(2*#1/(1+#1*#1+#2*#2))}
 \pgfmath@smuggleone\pgfmathresult%
\endgroup}

\makeatother

\newcommand{\x}{x}  
\newcommand{\y}{y}  
\renewcommand{\u}{u}  
\newcommand{\T}{\ensuremath{\mathbf{T}}}  
\newcommand{\A}{A}  
\newcommand{\B}{B}  
\newcommand{\C}{C}  
\newcommand{\D}{D}  
\newcommand{\m}{m}  
\renewcommand{\P}{P}  
\newcommand{\M}{M}  
\newcommand{\N}{N}  

\newcommand{\rv}[1]{\textcolor{black}{#1}}

\renewcommand{\Re}{\operatorname{Re}}
\renewcommand{\Im}{\operatorname{Im}}
\newcommand{\F}{\mathcal{F}_\infty}
\newcommand{\Hf}{\mathcal{H}_\infty}
\newcommand{\htwo}{\mathcal{H}_2}
\newcommand{\hilbert}{\mathcal{H}}

\newcommand{\Real}{\ensuremath{\mathbb{R}}}
\newcommand{\Comp}{\ensuremath{\mathbb{C}}}
\newcommand{\Compex}{\ensuremath{\hat{\mathbb{C}}}}
\newcommand{\cl}[1]{\ensuremath{\mathrm{cl}\,{#1}}}
\newcommand{\id}{\ensuremath{\mathbf{I}}}

\newcommand{\srg}[1]{\ensuremath{\mathrm{SRG}\left(#1\right)}} 

\newcommand{\norm}[1]{\left\lVert#1\right\rVert}
\newcommand{\dom}[1]{\ensuremath{\mathcal{D}_{#1}}} 
\newcommand{\annulus}[2]{\ensuremath{\mathrm{Ann}(#1,#2)}}

\newcommand{\graph}[1]{\ensuremath{\mathcal{G}_{#1}}}
\newcommand{\maxgain}[1]{\ensuremath{\bar{\sigma}\left(#1\right)}}
\newcommand{\mingain}[1]{\ensuremath{\ubar{\sigma}\left(#1\right)}}
\newcommand{\ltwotau}{\ensuremath{\mathcal{L}_{2,\tau}^{m}}}
\newcommand{\ltwoinf}{\ensuremath{\mathcal{L}_{2}^{m}}} 
 
\newcommand{\trunc}{\ensuremath{\mathbf{P}_\tau}} 
\newcommand{\fbk}[1]{\ensuremath{f_{\mathrm{BK}}\! \left(#1\right)}} 
\newcommand{\gbk}[1]{\ensuremath{g_{\mathrm{BK}} \! \left(#1\right)}} 
\newcommand{\Tinf}{\ensuremath{\mathbf{T}_{T(s)}}}
\newcommand{\Tinfhat}{\ensuremath{\hat{\mathbf{T}}_{T(s)}}}
\newcommand{\Tinfs}{\ensuremath{\mathbf{T}_{T(s)}^\sigma}}
\newcommand{\Tinfshat}{\ensuremath{\hat{\mathbf{T}}_{T(s)}^\sigma}}
\newcommand{\Tinfinverse}{\ensuremath{\mathbf{T}_{T(s)^{-1}}}}
\newcommand{\Ttau}{\ensuremath{\mathbf{T}_{T(s)}^\tau}}
\newcommand{\Ttauinverse}{\ensuremath{\mathbf{T}_{T(s)^{-1}}^\tau}}

\newcommand{\esssup}{\ensuremath{\mathrm{ess}\,\sup_{\omega\in\mathbb{R}}\maxgain{T(i\omega)}}}
\newcommand{\essinf}{\ensuremath{\mathrm{ess}\,\inf_{\omega\in\mathbb{R}}\mingain{T(i\omega)}}}

\newcommand{\ubar}[1]{\underaccent{\bar}{#1}}

\author{Talitha Nauta and Richard Pates}
\thanks{The authors are with the Department of Automatic Control, Lund University, Box 118, SE 221 00 LUND, Sweden (e-mail: {\tt\{talitha.nauta, richard.pates\}@control.lth.se}).}
\thanks{The authors are members of the ELLIIT
Strategic Research Area at Lund University.}

\begin{document}

\begin{abstract}
     The \gls{srg} is a promising tool for stability and robustness analysis of multi-input multi-output systems. In this paper, we provide tools for exact and computable constructions of the \gls{srg} for closed linear operators, based on maximum and minimum gain computations. The results are suitable for bounded and unbounded operators, and we specify how they can be used to draw \glspl{srg} for the typical operators that are used to model linear-time-invariant dynamical systems. For the special case of state-space models, we show how the Bounded Real Lemma can be used to construct the \gls{srg}. \rv{Furthermore, we give an example of how stability analysis of dynamical systems can be done with \glspl{srg}.}
     \glsresetall
\end{abstract}

\maketitle

\section{Introduction}
\glspl{srg} provide a geometrical tool for the analysis of operators. The \gls{srg} was first introduced by Ryu, Hannah, and Yin in~\cite{Ryu2021} as a tool for providing and discovering rigorous proofs of convergence for optimisation algorithms. The results were later connected to classical systems theory in~\cite{Chaffey2021} and~\cite{Chaffey2023}. In this setup, diagrams based on \glspl{srg} were used to determine stability and robustness properties of feedback interconnections of linear and non-linear multi-input multi-output systems \cite{Chaffey2021,Chaffey2023,baronprada2025,baronprada2025a,chen2025a,krebbekx2025}. Other domains in the control literature where the application of \glspl{srg} has been introduced are in dominance analysis~\cite{chen2025}, the Lur'e problem~\cite{krebbekx2025a}, and integral quadratic constraints~\cite{chen2025a,degroot2025,Eijnden2024}. 

To unlock the full potential of the \gls{srg} it is crucial to know how to determine it. This has led to a significant effort from the research community to provide methods for obtaining, or obtaining over-approximations of, the \gls{srg} for a range of commonly encountered operators. However, even for the case of linear operators, the results remain incomplete. The case of obtaining the \gls{srg} for normal matrices (and two-by-two matrices) was first solved in~\cite{Huang2024}. It was then shown in~\cite{Pates21} that the \glspl{srg} for bounded linear operators on Hilbert spaces could be obtained through a certain mapping of the Numerical Range. Results tailored towards the types of \gls{lti} operators encountered in systems theory were first presented in~\cite{Chaffey2021,Chaffey2023}, where Nyquist like diagrams were connected to the construction of the \glspl{srg} of operators related to stable single-input single-output \gls{lti} systems. These results were then extended to normal operators and connected to tools from dissipativity theory in~\cite{degroot2025}, where further connections to non-linear operators were also given.

A difficulty in obtaining \gls{srg} constructions in the systems theory setting stems from the fact that it is typical to model systems using operators on extended spaces. This is often crucial when deducing stability properties of interconnections from their constituent parts, but is an obstruction to standard \gls{srg} definitions which are firmly grounded in the theory of operators on Hilbert spaces. Recently, this has led to the consideration of so called soft and hard \glspl{srg}~\cite{chen2025a}. These results increase compatibility with conventional integral quadratic constraints or dissipativity tools, but typically come at the price of sharpness, in the sense that \gls{srg} constructions are replaced with constructions of over-approximations of \glspl{srg}~\cite{degroot2025, krebbekx2025}.

In this paper, we aim to narrow this gap by giving exact and computable constructions of \glspl{srg} for the types of \gls{lti} operators typically encountered in the study of dynamical systems. Our main contribution is presented in Section~\ref{sec:mainresult}, where it is shown that the \gls{srg} of any linear, possibly unbounded, operator with closed graph can be determined exactly from maximum and minimum gain computations. We then specialise these results to operators obtained from a transfer function definition in Section~\ref{sec:operatorsonL2e}, where we include both unstable and multi-input multi-output transfer functions. We distinguish between two types of operators in subsections~\ref{sec:ltwoinf}~and~\ref{sec:ltwotau}, each of which is closely related to the notions of soft and hard \glspl{srg} in~\cite{chen2025a}. We further give frequency domain formulas for the corresponding gain calculations (presented as Theorems~\ref{thm:gain_T_inf} and \ref{thm:gain_T_tau} respectively), allowing the \glspl{srg} for both representations to be computed to any desired level of precision. Lastly, we connect the gain computations to the Bounded Real Lemma in Section~\ref{sec:ss}, allowing convex optimisation to be used to determine the corresponding \glspl{srg} whenever the transfer functions can be written on state-space form. Examples are given in Section~\ref{sec:examples}, where we also explain how the \gls{srg} can be used for stability analysis of dynamical systems.

\section{Preliminaries}
\label{sec:preliminaries}
\subsection{Basic notation}
The complex conjugate of $z \in \mathbb{C}$ is denoted $\bar{z}$ and we denote the conjugate transpose $(\cdot)^*$. The closure of a set $S$ is denoted $\cl{S}$. The \textit{extended complex plane} is the set $\Compex\coloneqq{}\mathbb{C} \cup \{\infty\}$, which we endow with the usual chordal metric. We define the limit of a sequence of subsets $S_k\subseteq{}\Compex$ as the closed set given by \[
\lim_{k\rightarrow{}\infty}S_k\coloneqq{}\{z\in\Compex:z=\lim_{k\rightarrow{}\infty}z_k,z_k\in S_k\}. \] 

A Hilbert space over the field $\mathbb{C}$ is denoted $\hilbert$. It is equipped with an inner product $\langle\cdot,\cdot\rangle: \hilbert~\times~\hilbert~\rightarrow~\mathbb{C}$ which defines a norm $\norm{\cdot} = \sqrt{\langle \cdot,\cdot \rangle}$. The Hilbert space of Lebesgue square integrable functions over a field $\mathbb{C}$ is the set of functions \[\ltwoinf \coloneqq \{ f: [0, \infty) \rightarrow \mathbb{C}^m : \norm{f} < \infty \},\] where the inner product is $\langle u, y \rangle = \int_0^\infty \u(t)^* \y(t) \, dt$. The Hilbert space $\ltwotau$ is the linear subspace of $\ltwoinf$ where $f(t) = 0$ when $t \in [\tau, \infty)$. The truncation operator $\trunc: \ltwoinf \rightarrow \ltwotau$ is defined such that $\trunc f(t) = f(t)$ on $[0, \tau]$ and zero for $t > \tau$. The Hardy space $\mathcal{H}_2^{m}$ denotes the set of functions $\hat{f}: \mathbb{C} \rightarrow \mathbb{C}^m$ such that
\begin{equation}
    \sup_{\sigma > 0} \left( \frac{1}{2\pi}\int_{-\infty}^{\infty} \hat{f}(\sigma + i\omega)^*\hat{f}(\sigma + i\omega) d\omega \right)^{\frac{1}{2}}< \infty, 
\end{equation}
and the Hardy space $\mathcal{H}_\infty^{m \times m}$ denotes the set of functions $\hat{f}: \mathbb{C} \rightarrow \mathbb{C}^{m \times m}$ such that
    $\sup_{\Re(s) > 0} \maxgain{\hat{f}(s)}< \infty,$
where $\maxgain{\cdot}$ denotes the largest singular value. By $\mathcal{F}_\infty^{m \times m}$ we denote the quotient field over $\mathcal{H}_\infty^{m \times m}$.

\subsection{Notation for closed linear operators}
We call $\T:\dom{\T} \subseteq \hilbert \rightarrow \hilbert$ a linear (possibly unbounded) operator if it is linear and $\dom{\T}$ is a linear manifold. The operator is \textit{closed} if its graph is a closed set in $\hilbert \times \hilbert$, 
where the graph of the operator is defined by
\begin{equation}
    \mathcal{G}_{\T} = \begin{pmatrix}
        I \\ \T
    \end{pmatrix} \dom{\T} \subset \hilbert \times \hilbert.
\end{equation}
We define a partial order over operators so that $\T_1\preceq\T_2$ if $\graph{\T_1}\subseteq{}\graph{\T_2}$. The identity operator is denoted $\id$.

\subsection{Scaled Relative Graphs}
We define the \gls{srg} of a possibly unbounded linear operator $\T:\dom{\T} \subseteq \hilbert \rightarrow \hilbert$, to be the subset of the extended complex plane given by
\[ \srg{\T} \!
    \coloneqq{} \!\!
    \left\{ \frac{\| \y \|}{\| \u\|} \exp \left(\! \pm i \angle(\u,\y) \right) \! : \!
    u \in \dom{\T} \backslash \{ 0 \} , \y = \T \u \right\}\!. \]
This extends the usual definition \cite{Ryu2021} in a natural way to cover operators with domain not necessarily equal to $\hilbert$. Note that the \textit{relative} part of the definition is not required, since the operators we consider are linear. Note also that while we only consider Hilbert spaces over $\Comp$, our results cover the real Hilbert space case through \cite[Theorem 2]{Pates21} (this result is easily extended to the closed operator setting considered here).

The \gls{srg} captures some geometric features of the input-output pairs of the operator $\T$. The first term $\|y\|/\|u\|$ represents the scaling of the output relative to the input and the exponent captures the angle between the input and output. The angle $\angle(u,y)$ between $u \in \mathcal{H}$ and $y \in \mathcal{H}$ is defined by its inner product
    \[ \cos{(\angle(u,y))} = \frac{\Re(\langle y,u \rangle)}{\| y \| \| u \|} \text{ where } \angle(u,y)=0 \text{ if } y = 0. \]
    Note that the \gls{srg} is symmetric about the real axis. 

\subsection{Beltrami-Klein mapping}
The Beltrami-Klein mapping maps the extended complex plane to the closed unit disc through
\[\fbk{z} \coloneqq{} \frac{(\Bar{z} - i)(z-i)}{1 + \bar{z}z}.\] The mapping is surjective, and bijective on the extended upper half-plane $\{z:\mathrm{Im}(z)\geq{}0\}\cup\{\infty\}$. Defining
\[
\gbk{z} \coloneqq{} \left\{ \frac{\Im(z) \pm i\sqrt{1 - |z|^2}}{\Re(z) - 1} \right\},
\]
it then follows that if $S$ is a subset of the unit disc, $S = \fbk{\gbk{S}}$. Furthermore, for any subset $S$ of the extended complex plane that is symmetric about the real axis $\left( z\in{}S\implies{}\bar{z}\in{}S\right)$, $S = \gbk{\fbk{S}}$.

The Beltrami-Klein mapping bijectively maps generalised circles centred on the extended real axis to chords of the unit circle. This translates convexity properties of subsets of the unit disc on to convexity properties of subsets of the extended complex plane with respect to the Poincar\'{e} metric. We will have particular need for the following lemma, which shows that the subsets of $\Compex$ given by the intersection of a set of bounding annuli have particular convexity properties. Each annulus is defined by the longest and shortest distance from a point $\alpha \in \mathbb{R}$ to a set $S \subseteq \Compex$, which are given by
\begin{align}
    d_l(S, \alpha) & = \sup \{|z-\alpha| : z \in S\} \\
    d_s(S, \alpha) & = \inf \{|z-\alpha| : z \in S\}. 
\end{align}
Using this, we define an annulus as the subset of the extended complex plane to be
\begin{equation}
    \label{eq:set_annuli}
    \annulus{S}{\alpha}
    \coloneqq{}
    \left\{ \alpha + z :d_s(S, \alpha) \leq |z| \leq d_l(S, \alpha) \right\},
\end{equation}
where it is understood that $\{\infty\} \in \annulus{S}{\alpha}$ if $d_l(S, \alpha) = \infty$ and $\annulus{S}{\alpha} = \{\infty\}$ if $d_s(S, \alpha) = d_l(S, \alpha) = \infty$. 
With this definition in place, we state the following result. 
\begin{lemma}
    \label{lem:annuli}
     A closed set $S\subseteq\{z\in\Comp:\vert{}z\vert{}\leq{}1\}$ is convex if and only if
    \begin{equation}\label{eq:seteq}
    \bigcap_{\alpha\in\Real}\annulus{\gbk{S}}{\alpha}=\gbk{S}.
    \end{equation}
\end{lemma}
\begin{proof}
    The Beltrami-Klein mapping bijectively maps circles centred on the extended real axis to chords of the unit circle. Hence
    \begin{equation}\label{eq:cvx}
    \fbk{\bigcap_{\alpha\in\Real}\annulus{\gbk{S}}{\alpha}}
    \end{equation}
    is equal to the intersection of a set of closed half-spaces with the closed unit disc. All these sets are convex, and so the set in \eqref{eq:cvx} is convex. Since $\fbk{\gbk{S}}=S$, this establishes that if \eqref{eq:seteq} holds, $S$ is convex. Now suppose that \eqref{eq:seteq} does not hold, with the goal of showing that $S$ is not convex. We proceed by contradiction, and so suppose that $S$ is convex. From the definition of the annuli, it is clear that for any $\alpha$, $\annulus{\gbk{S}}{\alpha}\supseteq{}\gbk{S}$. This implies that if \eqref{eq:seteq} does not hold, there exists a $z\notin{}S$ such that for all $\alpha$,
    \begin{equation}\label{eq:cont}
    z\in\fbk{\annulus{\gbk{S}}{\alpha}}.
    \end{equation}
    Since $S$ is assumed convex, there exists a hyperplane that strictly separates $z$ and $S$. However again, since the Beltrami-Klein mapping bijectively maps circles to lines, this implies that there exists an $\alpha\in\Real$ such that $z\notin{}\fbk{\annulus{\gbk{S}}{\alpha}}$, see Figure~\ref{hyperplanes} for an illustration. This contradicts \eqref{eq:cont}, and the proof is complete.
\end{proof}

\begin{figure}[h]
    \centering
    \begin{subfigure}{0.45\textwidth}     
        \centering
            \begin{tikzpicture}[every to/.style={hyperbolic plane},scale=0.9,>={Stealth[scale=1]}]
        
       \pgfmathsetmacro{\myalpha}{3}  
       \pgfmathsetmacro{\maxr}{2.8639}
       \pgfmathsetmacro{\minr}{1.7459}

        \let\radius\undefined
        \newlength{\radius}
        \setlength{\radius}{1.5pt}

        \draw[fill=gray!30,even odd rule] (\myalpha,0) circle (\maxr) (\myalpha,0) circle (\minr);
        \input{simple_example.tex}
  
        \coordinate (alpha) at (\myalpha,0);

        \fill (alpha) circle[radius=\radius];
        \node[below] at (\myalpha,-.1) {$\alpha$};

        \fill (-0.15,0.8) circle[radius=\radius];
        \node[left] at (-0.15,0.8) {$\gbk{z}$};

        \draw[->] (-0.9,0) -- (6.9,0) node[below]{$\Re$};
        \draw[->] (0,-3.1) -- (0,3.1) node[left]{$\Im$};
        \draw (\myalpha,.05) -- (\myalpha,-.15);
        \draw[->] (alpha) -- node[above,rotate=45] {$d_s(S, \alpha)$} ({\myalpha+\minr*cos(45)},{\minr*sin(45)});
        \draw[->] (alpha) -- node[above,rotate=-45] {\phantom{ } $d_l(S, \alpha)$} ({\myalpha+\maxr*cos(45)},{-\maxr*sin(45)});
    \end{tikzpicture}
        \caption{The areas in the extended complex plane, where the orange area is the set $\gbk{S}$ and the grey area is the set $\annulus{\gbk{S}}{\alpha}$.}
        \label{fig:annuli}
    \end{subfigure}
    
    \vspace{1em} 
    
    \begin{subfigure}{0.45\textwidth}
        \centering
            \begin{tikzpicture}[>={Stealth[scale=1]},scale=2.7]

       \pgfmathsetmacro{\myalpha}{3}  
       \pgfmathsetmacro{\maxr}{2.8639}
       \pgfmathsetmacro{\minr}{1.7459}

        \fill[gray!30] (0,0) circle (1);

        \begin{scope}
        \clip (0,0) circle (1);
        \draw[fill=white] ({bkx(\myalpha-\minr,0)},{bky(\myalpha-\minr,0)}) -- ({bkx(\myalpha+\minr,0)},{bky(\myalpha+\minr,0)}) -- (2,0) -- (-2,-10);
        \draw[fill=white] ({bkx(\myalpha-\maxr,0)},{bky(\myalpha-\maxr,0)}) -- ({bkx(\myalpha+\maxr,0)},{bky(\myalpha+\maxr,0)})-- (2,2) -- (-2,2);
        \end{scope}
	     
         \draw[fill=DutchOrange!80] (-0.1281352539, -0.3349716664) to (-0.1618172966, -0.3248472379) to (-0.1954404661, -0.3158868207) to (-0.2287287501, -0.3081316368) to (-0.2614254972, -0.3015944504) to (-0.2933021107, -0.2962611614) to (-0.3241642178, -0.2920939048) to (-0.3538551017, -0.2890351933) to (-0.3822565241, -0.2870125955) to (-0.4092873264, -0.2859434791) to (-0.4349003488, -0.2857394407) to (-0.4590782485, -0.2863101627) to (-0.4818287578, -0.2875665608) to (-0.5031798309, -0.2894231839) to (-0.52317501, -0.2917999069) to (-0.5418692303, -0.2946230027) to (-0.55932518, -0.2978257004) to (-0.575610258, -0.3013483457) to (-0.5907941139, -0.3051382661) to (-0.6049467248, -0.3091494351) to (-0.6181369404, -0.3133420083) to (-0.6304314259, -0.3176817896) to (-0.6418939298, -0.3221396717) to (-0.6525848117, -0.3266910794) to (-0.6625607726, -0.3313154372) to (-0.6718747387, -0.3359956741) to (-0.6805758582, -0.3407177703) to (-0.6887095802, -0.3454703504) to (-0.6963177884, -0.3502443225) to (-0.703438972, -0.3550325603) to (-0.7101084169, -0.3598296265) to (-0.7163584078, -0.3646315332) to (-0.7222184322, -0.3694355355) to (-0.7277153807, -0.3742399553) to (-0.7328737394, -0.379044031) to (-0.7377157729, -0.3838477904) to (-0.7422616948, -0.3886519433) to (-0.746529826, -0.3934577916) to (-0.7505367397, -0.398267155) to (-0.7542973937, -0.4030823087) to (-0.7578252501, -0.4079059331) to (-0.7611323827, -0.4127410734) to (-0.7642295725, -0.417591107) to (-0.7671263927, -0.4224597189) to (-0.7698312819, -0.4273508829) to (-0.7723516088, -0.4322688492) to (-0.7746937257, -0.4372181366) to (-0.7768630142, -0.4422035287) to (-0.778863921, -0.447230075) to (-0.7806999857, -0.4523030947) to (-0.7823738598, -0.4574281846) to (-0.7838873181, -0.4626112292) to (-0.7852412612, -0.4678584146) to (-0.7864357094, -0.4731762439) to (-0.7874697892, -0.478571556) to (-0.7883417098, -0.4840515465) to (-0.7890487313, -0.4896237905) to (-0.7895871222, -0.495296268) to (-0.7899521081, -0.5010773906) to (-0.7901378081, -0.5069760307) to (-0.7901371602, -0.513001551) to (-0.7899418336, -0.5191638362) to (-0.7895421273, -0.5254733244) to (-0.7889268544, -0.5319410394) to (-0.788083209, -0.5385786211) to (-0.7869966172, -0.545398355) to (-0.7856505686, -0.552413198) to (-0.7840264279, -0.559636799) to (-0.7821032261, -0.5670835129) to (-0.7798574295, -0.5747684037) to (-0.7772626871, -0.5827072344) to (-0.7742895562, -0.5909164384) to (-0.7709052083, -0.5994130668) to (-0.7670731177, -0.6082147053) to (-0.7627527381, -0.6173393516) to (-0.7578991756, -0.6268052429) to (-0.7524628692, -0.6366306226) to (-0.7463892954, -0.6468334308) to (-0.7396187205, -0.657430903) to (-0.7320860297, -0.6684390596) to (-0.7237206746, -0.6798720661) to (-0.7144467881, -0.6917414461) to (-0.7041835316, -0.7040551283) to (-0.6928457484, -0.716816316) to (-0.6803450111, -0.7300221729) to (-0.666591159, -0.7436623351) to (-0.6514944214, -0.7577172771) to (-0.6349682167, -0.772156588) to (-0.6169326866, -0.7869372482) to (-0.5973189818, -0.8020020386) to (-0.5760742354, -0.8172782557) to (-0.55316706, -0.8326769467) to (-0.5285932731, -0.8480929019) to (-0.5023814187, -0.8634056388) to (-0.4745975221, -0.87848157) to (-0.4453484332, -0.8931774567) to (-0.4147831135, -0.9073451093) to (-0.3830913266, -0.9208371192) to (-0.3504994324, -0.9335132268) to (-0.3172633163, -0.9452467721) to (-0.2836588784, -0.9559306008) to (-0.2499708653, -0.9654818138) to (-0.2164810856, -0.9738448819) to (-0.1834571276, -0.9809928572) to (-0.1511425982, -0.9869266679) to (-0.1197496375, -0.9916727232) to (-0.08945410925, -0.9952792314) to (-0.06039350481, -0.9978117271) to (-0.03266728782, -0.9993483068) to (-0.006339203917, -0.9999750019) to (0.01855901779, -0.9997816094) to (0.0420231412, -0.998858179) to (0.06407157684, -0.9972922383) to (0.08474060391, -0.9951667542) to (0.1040799346, -0.9925587623) to (0.1221487862, -0.9895385649) to (0.1390125379, -0.9861693852) to (0.154739983, -0.9825073663) to (0.1694011437, -0.9786018166) to (0.1830655907, -0.9744956179) to (0.195801195, -0.9702257302) to (0.2076732419, -0.9658237429) to (0.2187438359, -0.9613164361) to (0.2290715364, -0.9567263269) to (0.2387111694, -0.9520721847) to (0.2477137728, -0.9473695061) to (0.2561266373, -0.9426309449) to (0.2639934154, -0.9378666958) to (0.2713542761, -0.9330848342) to (0.2782460872, -0.9282916122) to (0.2847026132, -0.923491717) to (0.2907547193, -0.9186884931) to (0.296430574, -0.9138841334) to (0.5294117647, -0.7058823529) to (0.08707595951, -0.4299197787) to (0.06006484443, -0.4147084214) to (0.03153004244, -0.3998047761) to (0.001599517059, -0.3853552242) to (-0.0295569633, -0.3715083312) to (-0.06173056863, -0.3584087351) to (-0.09467907849, -0.3461906617) to cycle; 

        \fill ({bkx(\myalpha,0)}, {bky(\myalpha,0)}) circle[radius=0.5pt];
        \node[right] at ({bkx(\myalpha,0)}, {bky(\myalpha,0)}) {$\fbk{\alpha}$};

        \fill ({bkx(-0.15,0.8)}, {bky(-0.15,0.8)}) circle[radius=0.5pt];
        \node[left] at ({bkx(-0.15,0.8)}, {bky(-0.15,0.8)}) {$z$};

	     \draw[dotted] (0,0) circle [radius=1];

        \draw[->] (-1.2,0) -- (1.2,0) node[below]{$\Re$};;
        \draw[->] (0,-1.1) -- (0,1.1) node[left]{$\Im$};;

        \node[below left] at (1,-0.05) {$1$};
        \draw (1,0) -- (1,-0.05);
        \node[below left] at (-1,-0.05) {$-1$};
        \draw (-1,0) -- (-1,-0.05);
        
        \node[below left] at (-0.05,1) {$1$};
        \draw (0,1) -- (-0.05,1);
        \node[below left] at (-0.05,-1) {$-1$};
        \draw (0,-1) -- (-0.05,-1);

\end{tikzpicture}
        \caption{The areas in the unit disc under Beltrami-Klein mapping, where the orange area is the set $S$ and the grey area is the set $\fbk{\annulus{\gbk{S}}{\alpha}}$.}
        \label{fig:hyperplane}
    \end{subfigure}
    \caption{The set $\gbk{S}$ and the bounding annulus with centre point $\alpha$ on the extended complex plane~(a) and under Beltrami-Klein mapping~(b).} 
    \label{hyperplanes}
\end{figure}

\section{Gain-based characterisation of the SRG}
\label{sec:mainresult}
In the following section, we show that the limit of the \glspl{srg} of a monotone sequence of closed linear operators $\T_k:\dom{\T_k} \subseteq \hilbert \rightarrow \hilbert$ can be obtained from the intersection of a set of annuli defined in terms of the maximum and minimum gains of operators. Our result covers the case of a single closed linear operator as a special case. This characterisation is contingent on the following result concerning the convexity properties of the \gls{srg}.

\begin{theorem}
    \label{lemma:convex}
    For a closed linear operator $\T$, $\fbk{{\srg{\T}}}$ is a convex set in the unit disc $\{z\in\mathbb{C}:\vert{}z\vert{}\leq1\}$.
\end{theorem}
\begin{proof} 
 Since $\T$ is closed, the graph of the operator $\graph{\T}$ is a closed subspace of $\hilbert\times \hilbert$ and therefore itself a Hilbert space because every closed subspace of a Hilbert space is a Hilbert space with the inner product inherited from $\hilbert$ (see e.g. \cite[Theorem 3.2-3]{kreyszig1991}). From \cite{Pates21} we know that the set $\fbk{\srg{\T}}$ is given by
 \begin{align}
        & \left\{ \frac{\langle R(\u,\y), (\u,\y) \rangle_{\hilbert\times \hilbert}}{\langle (\u,\y), (\u,\y) \rangle_{\hilbert\times \hilbert}} \;\middle|\; (\u,\y) \in \graph{\T} = \Tilde{\hilbert} \right\}  = \! \left\{ \frac{\langle Rz, z \rangle_{\Tilde{\hilbert}}}{\langle z, z \rangle_{\Tilde{\hilbert}}} \;\middle|\; z \in \Tilde{\hilbert} \right\},
    \end{align}
    where $R(u,y) = (-i\y - \u, \y - i\u)$ and the last set is the numerical range over $\Tilde{\hilbert}$. As the numerical range is a convex set and $|\langle Rz, z \rangle |\leq \langle z, z \rangle$ by Cauchy-Schwarz, $\fbk{\srg{\T}}$ is a convex set inside the unit disc. 
\end{proof}

Lemma~\ref{lem:annuli} can be combined with Theorem~\ref{lemma:convex} to show that the closure of the \gls{srg} of a closed operator can be computed from the intersection of a set of annuli. These computations can be further simplified in terms of the maximum and minimum gains of an operator, which we define according to
\begin{align}
      \maxgain{\T}
   \coloneqq{}
   \sup_{\substack{\u \in \dom{\T} \\ \u \neq 0}} \frac{\norm{\T\u}}{\norm{\u}}
 \;\;\text{  and  }\;\;
      \mingain{\T}
   \coloneqq{}
   \inf_{\substack{\u \in \dom{\T} \\ \u \neq 0}}  \frac{\norm{\T \u}}{\norm{\u}}.
\end{align}
In a slight abuse of notation, we now define the annulus of an operator according to
\begin{equation}
    \label{eq:operator_annulus}
    \annulus{\T}{\alpha}
    \coloneqq{}
    \left\{ \alpha + z :\mingain{\T - \alpha \id} \leq |z| \leq \maxgain{\T - \alpha \id} \right\}.
\end{equation}
From the definition of the \gls{srg} and the fact that $\srg{\T - \alpha \id}=\srg{\T}-\alpha$, (see Figure~\ref{fig:annuli} and consider the case $S = \fbk{\srg{\T}}$), it follows that 
\[
\annulus{\srg{\T}}{\alpha}=\annulus{\T}{\alpha}.
\]
That is, this notation is consistent with the earlier definition in \eqref{eq:set_annuli} if the set in question is $\srg{\T}$. 

We now state our main result, which shows that the limit of the \glspl{srg} of a monotone sequence of closed operators can be computed from the intersection of a set of annuli. Since the annuli can be computed from the notions of maximum and minimum gain, this gives a way to compute the \gls{srg} using standard algorithms. We shall give several examples of this in the coming section. Note that we state the result in terms of a well-ordered sequence of operators for a technical reason connected to the study of operators on extended spaces. The presented result can always be used to characterise the \gls{srg} of a single closed linear operator $\T$, as made specific in Corollary~\ref{cor:ann} below.  
\begin{theorem}
    \label{circle_thm}
    Let $\T_{0}\preceq{}\T_1\preceq{}\T_2\preceq{}\ldots{}$ be a sequence of closed linear operators $\T_k:\dom{\T_k}\subseteq{}\hilbert\rightarrow{}\hilbert$. Then
    \begin{align}
        \lim_{k\rightarrow{}\infty}\srg{\T_k} = \bigcap_{\alpha \in \Real} \lim_{k\rightarrow\infty}\annulus{\T_k}{\alpha}.
    \end{align}
\end{theorem}
\begin{proof}
    From Theorem~\ref{lemma:convex}, $S_k\coloneqq{}\fbk{\srg{\T_k}}$ is convex and lies inside the unit disc. Since the Beltrami-Klein mapping is bijective on the upper extended half-plane of the complex plane, and $\gbk{\cdot}$ is a continuous mapping between the unit disc and the upper extended half-plane (with respect to the chordal metric), it follows that
    \[  
    \begin{aligned} \lim_{k\rightarrow{}\infty}\srg{\T_k}&= \lim_{k\rightarrow{}\infty}\gbk{S_k}= \gbk{\lim_{k\rightarrow\infty}S_k}.
    \end{aligned}
    \]
     Furthermore since $\T_k\preceq{}\T_{k+1}$ for all $k$, it follows that $S_0\subseteq{}S_1\subseteq{}S_2\ldots{}$, which then implies that
    \[
    S_\infty\coloneqq{}\lim_{k\rightarrow{}\infty}S_k
    \]
    exists and is a bounded closed subset of the unit disc.  This establishes that
    \[
    \lim_{k\rightarrow{}\infty}\srg{\T_k}=\gbk{S_\infty}.
    \]
    Now consider any two points in the interior of $S_\infty$. Then for sufficiently large $k$, these points lie in the set $S_k$. Since $S_k$ is convex, the line between the two points also lies in $S_k$. This shows that the interior of $S_\infty$ is convex, and it follows that $S_\infty$ is a convex set. Then Lemma~\ref{lem:annuli} gives that 
        \begin{align}
        \lim_{k\rightarrow{}\infty}\srg{\T_k} = \bigcap_{\alpha \in \Real} \annulus{\gbk{S_\infty}}{\alpha}.
    \end{align}
    For each $\alpha\in\Real$, $\annulus{\gbk{S_\infty}}{\alpha}$ is defined by the longest and shortest distance to the set $\gbk{S_\infty}$. Since $\gbk{S_k}=\srg{\T_k}$, the longest and shortest distances from $\alpha$ to $\srg{\T_k}$ are given by $\maxgain{\T_k-\alpha{}\id}$ and $\mingain{\T_k-\alpha{}\id}$ respectively. It then follows that the longest and shortest distances from $\alpha$ to $\gbk{S_\infty}$ are given by
    \begin{equation}
        \label{eq:lim_gain}
        \lim_{k\rightarrow{}\infty}\maxgain{\T_k-\alpha{}\id}
        \;\text{and}\;\lim_{k\rightarrow{}\infty}\mingain{\T_k-\alpha{}\id},
    \end{equation} 
    respectively. Hence
    \[
    \annulus{\gbk{S_\infty}}{\alpha}=\lim_{k\rightarrow\infty}\annulus{\T_k}{\alpha},
    \]
    and the proof is complete.
\end{proof}
To obtain the \gls{srg} for a single closed linear operator $\T$ we can set $\T_k=\T$ for all $k$ in Theorem~\ref{circle_thm}. This yields the following special case of Theorem~\ref{circle_thm}.
\begin{corollary}
    \label{cor:ann}
    Let $\T:\dom{\T}\subseteq{}\hilbert\rightarrow{}\hilbert$ be a closed linear operator. Then
    \begin{align}
        \cl{\srg{\T}} = \bigcap_{\alpha \in \Real} \annulus{\T}{\alpha}.
    \end{align}
\end{corollary}

\section{Computation of SRGs for dynamical systems}
\label{sec:operatorsonL2e}

Theorem~\ref{circle_thm} shows that if we can compute 
\begin{equation}\label{eq:neededgain}
\lim_{k\rightarrow{}\infty}\maxgain{\T_k-\alpha\id{}}\;\text{and}\;\lim_{k\rightarrow{}\infty}\mingain{\T_k-\alpha\id{}},
\end{equation}
we can compute the closed set $\lim_{k\rightarrow{}\infty}\srg{\T_k}$ to any desired level of precision by gridding over $\alpha$. The basic procedure is illustrated in Figure~\ref{fig:bounding_ann}, and summarised in Algorithm~\ref{alg:outer}. In short, the algorithm proceeds over a grid of points $\{\alpha_1, \alpha_2, \dots, \alpha_n \}$, and inserts the computations in \eqref{eq:neededgain} into \texttt{MaxGain($\alpha$)} and \texttt{MinGain($\alpha$)}, respectively. The algorithm returns these numbers for each $\alpha_i$, which is all that is required to specify $\lim_{k\rightarrow{}\infty}\annulus{\T_k}{\alpha_i}$.

\begin{algorithm}
\caption{Computation of the SRG}\label{alg:outer}
\DontPrintSemicolon
\SetKwFunction{SRGfromAnn}{SRGfromAnn}
\SetKwFunction{Grid}{Grid}
\SetKwFunction{MaxGain}{MaxGain}
\SetKwFunction{MinGain}{MinGain}

\SetKwProg{Fn}{Function}{}{}
\Fn{\SRGfromAnn{$n$}}{
    $\alpha \gets$ \Grid{$n$}\;
    $\ubar{\sigma} \gets$ vector of length $n$\;
    $\bar{\sigma} \gets$ vector of length $n$\;
    \For{$i \gets 1$ \KwTo $n$}{
        $\ubar{\sigma}(i) \gets$ \MinGain{$\alpha(i)$}\;
        $\bar{\sigma}(i) \gets$ \MaxGain{$\alpha(i)$}\;
    }
    \KwRet{$\alpha, \ubar{\sigma}, \bar{\sigma}$}\;
}
\end{algorithm}
\begin{figure}
    \centering
    \begin{tikzpicture}[every to/.style={hyperbolic plane},scale=0.7,>={Stealth[scale=1]}]  
       \pgfmathsetmacro{\myalphaone}{3}  
       \pgfmathsetmacro{\maxrone}{2.8639}
       \pgfmathsetmacro{\minrone}{1.7459}
       \pgfmathsetmacro{\myalphatwo}{2}  
       \pgfmathsetmacro{\maxrtwo}{1.9053}
       \pgfmathsetmacro{\minrtwo}{0.7873}
       \pgfmathsetmacro{\myalphathree}{-1}  
       \pgfmathsetmacro{\maxrthree}{2.6926}
       \pgfmathsetmacro{\minrthree}{1.2610}

        \let\radius\undefined
        \newlength{\radius}
        \setlength{\radius}{1.5pt}

        \begin{scope}
          \clip (\myalphaone,0) circle (\maxrone);
          \clip (\myalphatwo,0) circle (\maxrtwo);
          \clip (\myalphathree,0) circle (\maxrthree);

          \fill[black!20] (-5,-5) rectangle (5,5);

          \fill[white] (\myalphaone,0) circle (\minrone);
          \fill[white] (\myalphatwo,0) circle (\minrtwo);
          \fill[white] (\myalphathree,0) circle (\minrthree);

        \end{scope}

         \draw[thick, green, even odd rule] (\myalphaone,0) circle (\maxrone) (\myalphaone,0) circle (\minrone);
        \draw[thick, red, even odd rule] (\myalphatwo,0) circle (\maxrtwo) (\myalphatwo,0) circle (\minrtwo);
        \draw[thick, blue, even odd rule] (\myalphathree,0) circle (\maxrthree) (\myalphathree,0) circle (\minrthree);

        \draw[fill=DutchOrange!80] (0.2969250941, 0.8274494306) to (0.279602687, 0.8020370307) to (0.2642430381, 0.7766589489) to (0.2507727086, 0.7515394831) to (0.239090181, 0.7268725563) to (0.2290734384, 0.7028189276) to (0.2205873719, 0.6795056246) to (0.2134904932, 0.657027224) to (0.2076406155, 0.6354485071) to (0.2028993476, 0.6148079969) to (0.1991353901, 0.5951219298) to (0.1962267363, 0.5763883003) to (0.1940619382, 0.5585907113) to (0.1925406248, 0.541701859) to (0.19157346, 0.5256865622) to (0.1910817058, 0.5105043057) to (0.1909965312, 0.496111311) to (0.1912581771, 0.4824621753) to (0.1918150585, 0.4695111312) to (0.1926228642, 0.4572129863) to (0.1936436901, 0.4455237992) to (0.1948452321, 0.4344013439) to (0.1962000504, 0.4238054084) to (0.1976849098, 0.4136979659) to (0.1992801964, 0.40404325) to (0.2009694066, 0.3948077593) to (0.2027387033, 0.3859602126) to (0.2045765325, 0.377471468) to (0.2064732946, 0.3693144211) to (0.2084210624, 0.3614638879) to (0.2104133416, 0.3538964815) to (0.2124448667, 0.3465904869) to (0.2145114282, 0.3395257359) to (0.216609726, 0.3326834866) to (0.2187372469, 0.3260463073) to (0.2208921599, 0.3195979668) to (0.2230732297, 0.313323331) to (0.2252797437, 0.3072082661) to (0.2275114518, 0.3012395488) to (0.2297685159, 0.2954047826) to (0.2320514699, 0.2896923203) to (0.2343611857, 0.2840911924) to (0.2366988478, 0.278591041) to (0.2390659325, 0.2731820584) to (0.2414641934, 0.2678549304) to (0.2438956509, 0.2626007839) to (0.2463625866, 0.2574111384) to (0.2488675408, 0.2522778606) to (0.2514133148, 0.2471931224) to (0.2540029754, 0.2421493618) to (0.2566398638, 0.237139246) to (0.2593276069, 0.2321556372) to (0.2620701329, 0.2271915602) to (0.2648716891, 0.2222401721) to (0.2677368641, 0.2172947335) to (0.2706706128, 0.2123485815) to (0.2736782861, 0.2073951045) to (0.2767656639, 0.2024277178) to (0.279938993, 0.1974398412) to (0.2832050295, 0.1924248783) to (0.2865710865, 0.1873761971) to (0.2900450877, 0.1822871128) to (0.2936356269, 0.1771508728) to (0.2973520343, 0.1719606439) to (0.3012044509, 0.1667095035) to (0.3052039102, 0.1613904336) to (0.3093624294, 0.1559963207) to (0.3136931103, 0.1505199605) to (0.3182102499, 0.144954071) to (0.3229294629, 0.1392913137) to (0.3278678153, 0.1335243273) to (0.3330439704, 0.1276457744) to (0.3384783465, 0.121648406) to (0.3441932873, 0.1155251479) to (0.3502132422, 0.1092692126) to (0.3565649564, 0.1028742435) to (0.363277667, 0.09633449844) to (0.3703833003, 0.08964507889) to (0.3779166637, 0.08280221614) to (0.3859156232, 0.07580362265) to (0.3944212517, 0.06864892052) to (0.4034779329, 0.061340158) to (0.4131333951, 0.05388242526) to (0.4234386486, 0.04628457844) to (0.4344477878, 0.03856007823) to (0.4462176168, 0.03072794282) to (0.4588070461, 0.02281380615) to (0.4722762071, 0.01485105912) to (0.4866852249, 0.006882033153) to (0.5020925988, 0.00104083708) to (0.5185531476, 0.00885396189) to (0.5361155075, 0.01648177498) to (0.5548192033, 0.02383628797) to (0.574691372, 0.030817186) to (0.5957432838, 0.03731263292) to (0.6179668765, 0.04320093363) to (0.6413315939, 0.04835315037) to (0.6657818623, 0.05263668471) to (0.6912355566, 0.05591972459) to (0.7175837666, 0.05807632794) to (0.7446920805, 0.05899179421) to (0.772403454, 0.05856788777) to (0.8005425595, 0.05672744649) to (0.8289213308, 0.05341794653) to (0.8573452754, 0.04861370002) to (0.885620044, 0.04231651612) to (0.9135577377, 0.03455483015) to (0.9409824962, 0.02538146727) to (0.9677350281, 0.01487033078) to (0.9936758878, 0.003112374821) to (1.018687448, 0.009788764577) to (1.042674643, 0.02372113891) to (1.065564645, 0.0385685808) to (1.087305695, 0.05421432528) to (1.107865311, 0.0705440432) to (1.12722811, 0.08744824237) to (1.145393433, 0.1048240481) to (1.162372934, 0.1225764143) to (1.178188254, 0.1406188392) to (1.19286886, 0.1588736709) to (1.206450102, 0.177272089) to (1.218971508, 0.1957538414) to (1.23047533, 0.2142668071) to (1.241005323, 0.232766445) to (1.250605744, 0.2512151761) to (1.259320551, 0.2695817361) to (1.267192767, 0.2878405277) to (1.274264002, 0.3059709908) to (1.280574089, 0.3239570055) to (1.286160831, 0.341786336) to (1.291059822, 0.3594501195) to (1.295304344, 0.3769424024) to (1.298925308, 0.3942597227) to (1.5, 1) to (0.4709261227, 0.9843735862) to (0.4412096079, 0.965991946) to (0.4128210411, 0.9458815422) to (0.385972594, 0.9242452683) to (0.3608429105, 0.9013187756) to (0.3375703269, 0.8773615835) to (0.3162485413, 0.8526466674) to cycle; \begin{scope}[yscale=-1] \draw[fill=DutchOrange!80] (0.2969250941, 0.8274494306) to (0.279602687, 0.8020370307) to (0.2642430381, 0.7766589489) to (0.2507727086, 0.7515394831) to (0.239090181, 0.7268725563) to (0.2290734384, 0.7028189276) to (0.2205873719, 0.6795056246) to (0.2134904932, 0.657027224) to (0.2076406155, 0.6354485071) to (0.2028993476, 0.6148079969) to (0.1991353901, 0.5951219298) to (0.1962267363, 0.5763883003) to (0.1940619382, 0.5585907113) to (0.1925406248, 0.541701859) to (0.19157346, 0.5256865622) to (0.1910817058, 0.5105043057) to (0.1909965312, 0.496111311) to (0.1912581771, 0.4824621753) to (0.1918150585, 0.4695111312) to (0.1926228642, 0.4572129863) to (0.1936436901, 0.4455237992) to (0.1948452321, 0.4344013439) to (0.1962000504, 0.4238054084) to (0.1976849098, 0.4136979659) to (0.1992801964, 0.40404325) to (0.2009694066, 0.3948077593) to (0.2027387033, 0.3859602126) to (0.2045765325, 0.377471468) to (0.2064732946, 0.3693144211) to (0.2084210624, 0.3614638879) to (0.2104133416, 0.3538964815) to (0.2124448667, 0.3465904869) to (0.2145114282, 0.3395257359) to (0.216609726, 0.3326834866) to (0.2187372469, 0.3260463073) to (0.2208921599, 0.3195979668) to (0.2230732297, 0.313323331) to (0.2252797437, 0.3072082661) to (0.2275114518, 0.3012395488) to (0.2297685159, 0.2954047826) to (0.2320514699, 0.2896923203) to (0.2343611857, 0.2840911924) to (0.2366988478, 0.278591041) to (0.2390659325, 0.2731820584) to (0.2414641934, 0.2678549304) to (0.2438956509, 0.2626007839) to (0.2463625866, 0.2574111384) to (0.2488675408, 0.2522778606) to (0.2514133148, 0.2471931224) to (0.2540029754, 0.2421493618) to (0.2566398638, 0.237139246) to (0.2593276069, 0.2321556372) to (0.2620701329, 0.2271915602) to (0.2648716891, 0.2222401721) to (0.2677368641, 0.2172947335) to (0.2706706128, 0.2123485815) to (0.2736782861, 0.2073951045) to (0.2767656639, 0.2024277178) to (0.279938993, 0.1974398412) to (0.2832050295, 0.1924248783) to (0.2865710865, 0.1873761971) to (0.2900450877, 0.1822871128) to (0.2936356269, 0.1771508728) to (0.2973520343, 0.1719606439) to (0.3012044509, 0.1667095035) to (0.3052039102, 0.1613904336) to (0.3093624294, 0.1559963207) to (0.3136931103, 0.1505199605) to (0.3182102499, 0.144954071) to (0.3229294629, 0.1392913137) to (0.3278678153, 0.1335243273) to (0.3330439704, 0.1276457744) to (0.3384783465, 0.121648406) to (0.3441932873, 0.1155251479) to (0.3502132422, 0.1092692126) to (0.3565649564, 0.1028742435) to (0.363277667, 0.09633449844) to (0.3703833003, 0.08964507889) to (0.3779166637, 0.08280221614) to (0.3859156232, 0.07580362265) to (0.3944212517, 0.06864892052) to (0.4034779329, 0.061340158) to (0.4131333951, 0.05388242526) to (0.4234386486, 0.04628457844) to (0.4344477878, 0.03856007823) to (0.4462176168, 0.03072794282) to (0.4588070461, 0.02281380615) to (0.4722762071, 0.01485105912) to (0.4866852249, 0.006882033153) to (0.5020925988, 0.00104083708) to (0.5185531476, 0.00885396189) to (0.5361155075, 0.01648177498) to (0.5548192033, 0.02383628797) to (0.574691372, 0.030817186) to (0.5957432838, 0.03731263292) to (0.6179668765, 0.04320093363) to (0.6413315939, 0.04835315037) to (0.6657818623, 0.05263668471) to (0.6912355566, 0.05591972459) to (0.7175837666, 0.05807632794) to (0.7446920805, 0.05899179421) to (0.772403454, 0.05856788777) to (0.8005425595, 0.05672744649) to (0.8289213308, 0.05341794653) to (0.8573452754, 0.04861370002) to (0.885620044, 0.04231651612) to (0.9135577377, 0.03455483015) to (0.9409824962, 0.02538146727) to (0.9677350281, 0.01487033078) to (0.9936758878, 0.003112374821) to (1.018687448, 0.009788764577) to (1.042674643, 0.02372113891) to (1.065564645, 0.0385685808) to (1.087305695, 0.05421432528) to (1.107865311, 0.0705440432) to (1.12722811, 0.08744824237) to (1.145393433, 0.1048240481) to (1.162372934, 0.1225764143) to (1.178188254, 0.1406188392) to (1.19286886, 0.1588736709) to (1.206450102, 0.177272089) to (1.218971508, 0.1957538414) to (1.23047533, 0.2142668071) to (1.241005323, 0.232766445) to (1.250605744, 0.2512151761) to (1.259320551, 0.2695817361) to (1.267192767, 0.2878405277) to (1.274264002, 0.3059709908) to (1.280574089, 0.3239570055) to (1.286160831, 0.341786336) to (1.291059822, 0.3594501195) to (1.295304344, 0.3769424024) to (1.298925308, 0.3942597227) to (1.5, 1) to (0.4709261227, 0.9843735862) to (0.4412096079, 0.965991946) to (0.4128210411, 0.9458815422) to (0.385972594, 0.9242452683) to (0.3608429105, 0.9013187756) to (0.3375703269, 0.8773615835) to (0.3162485413, 0.8526466674) to cycle; \end{scope}


        \fill (\myalphaone,0) circle[radius=\radius];
        \fill (\myalphatwo,0) circle[radius=\radius];
        \fill (\myalphathree,0) circle[radius=\radius];
        \node[below, green] at (\myalphaone,-.1) {$\alpha_1$};
        \node[below, red] at (\myalphatwo,-.1) {$\alpha_2$};
        \node[below, blue] at (\myalphathree,-.1) {$\alpha_3$};
        
        \draw[->] (-4,0) -- (7,0) node[below]{$\Re$};
        \draw[->] (0,-3.5) -- (0,3.5) node[left]{$\Im$};
        \draw (\myalphaone,.05) -- (\myalphaone,-.15);
        \draw (\myalphatwo,.05) -- (\myalphatwo,-.15);
        \draw (\myalphathree,.05) -- (\myalphathree,-.15);
    \end{tikzpicture}
    \caption{Illustration of Algorithm~\ref{alg:outer} for computation of the \gls{srg}. The orange area shows the \gls{srg} while the grey area shows the approximation that is given by the intersection of the annuli defined by $\{ \alpha_1, \alpha_2, \alpha_3\}$.}
    \label{fig:bounding_ann}
\end{figure}

The grid over $\alpha$, defined by \texttt{Grid}($n$) in the algorithm, can be specified in many different ways. However, the computational efficiency for a desired level of precision can be improved significantly by choosing grid points judiciously.
One straightforward approach is to take equally spaced points over an interval on the real axis. In practice, it has proven more efficient to instead construct a grid of equally spaced $\fbk{\alpha}$ on the unit circle and then map these back to the real axis using $\gbk{\cdot}$. This method is inspired by the connection of the \gls{srg} to the numerical range described in~\cite{Pates21} and the insights on computing the numerical range in~\cite[Chapter~1.5]{Horn_Johnson_1991}.

Algorithm~\ref{alg:outer} can be used in the construction of the \gls{srg} for different types of operator sequences, by specifying the appropriate method to compute the maximum and minimum gain. The simplest case is the setting in Corollary~\ref{cor:ann} (all elements of the sequence are equal), where in addition the operators are matrices in $\Comp^{\m\times{}m}$. Then the required maximum and minimum gain computations correspond to calculating the maximum and minimum singular values of a matrix.
In the remainder of this section, we explain how to compute the maximum and minimum gains for two types of operators that describe the dynamics of \gls{lti} systems. 
The operators we consider can all be described by transfer functions $T(s)\in\F^{\m\times{}\m}$. This class is very broad, allowing for any $T(s)$ such that
\[
T=NM^{-1},\;\text{where}\;\begin{pmatrix}
    M\\N
\end{pmatrix}\in\Hf^{2\m\times{}\m}.
\]
This class therefore includes all transfer functions associated with state-space models, along with a host of other less frequently encountered transfer functions such as
\[
s,\;e^{-s},\;\frac{1}{1-e^{-s}}\;\text{and}\;e^s.
\]
To associate transfer functions with closed operators, we follow \cite{Georgiou1992}. In Section~\ref{sec:ltwoinf} we consider the case of operators defined on $\ltwoinf$, and present the relevant formula for the maximum and minimum gains in Theorem~\ref{thm:gain_T_inf}. We then further specialise to the case of operators acting on truncations of $\ltwoinf$ in Section~\ref{sec:ltwotau}, where Theorem~\ref{thm:gain_T_tau} gives the appropriate formulas. Finally, in Section~\ref{sec:ss} we explain how convex optimisation techniques may be used to compute the maximum and minimum gains of the associated operators whenever $T(s)$ is proper and real-rational.

\subsection{\gls{lti} operators defined on $\ltwoinf$}\label{sec:ltwoinf}

A transfer function $T=NM^{-1}\in\F^{\m\times{}\m}$ naturally defines an operator
\begin{equation}
    \hat{\T}: \dom{\hat{\T}} \subseteq \mathcal{H}_2^m \rightarrow \mathcal{H}_2^m
\end{equation}
with graph
\[
\graph{\hat{\T}}=\begin{pmatrix}
    M\\N
\end{pmatrix}\htwo^{\m}.
\]
Since the Fourier transform gives a Hilbert space isomorphism between $\ltwoinf$ and $\htwo^{\m}$, we can associate $\hat{\T}$ with an operator
\begin{equation}
    \T: \dom{\T} \subseteq \ltwoinf \rightarrow \ltwoinf
\end{equation}
defined by the relation $\hat{\T} \hat{u} = T\hat{u}=(\widehat{\T u})$. That is, the Fourier transform of every output of $\T$ is equal to the Fourier transform of the input multiplied by the transfer function. The operator $\T$ thus captures what is normally meant by an \gls{lti} transfer function model in the time domain, and the fact that the operators may be unbounded ensures that the analysis is not limited to stable systems.

A complication here is that without further restrictions, the above constructions may not yield a unique operator. For example, given any $T=NM^{-1}$ we may equally well define two operators $\hat{\T}_1$ and $\hat{\T}_2$ via
\[
\graph{\hat{\T}_1}=\begin{pmatrix}
M\\
N
\end{pmatrix}\htwo^{\m}\;\text{and}\;\graph{\hat{\T}_2}=\begin{pmatrix}
Me^{-s}\\
Ne^{-s}
\end{pmatrix}\htwo^{\m}.
\]
Both are associated with the same transfer function, however it follows from \cite[Proposition 6]{Georgiou1992} that $\graph{\hat{\T}_1}\supset{}\graph{\hat{\T}_2}$ (the inclusion is strict). Hence, the operators are partially ordered ($\T_1\succeq{}\T_2$ if $\graph{\T_1}\supseteq{}\graph{\T_2}$), which suggests defining the notion of a maximal operator associated with a transfer function.

\begin{definition}\label{def:ltwoinf}
Let $T\in\F^{\m\times{}\m}$ and define $\mathcal{T}$ as the set \[
\left\{\T:\graph{\hat{\T}}=
\begin{pmatrix}
    M\\N
\end{pmatrix}\htwo^{\m},\begin{pmatrix}
    M\\N
\end{pmatrix}\in\Hf^{2m\times{}m},T = NM^{-1}\right\}.
\]
Whenever it exists, define $\Tinf$ to be the maximum element of $\mathcal{T}$
(that is $\Tinf$ is the operator $\Tinf\in\mathcal{T}$ such that $\Tinf\succeq{}\T$ for all $\T\in\mathcal{T}$).
\end{definition}
Maximum elements, whenever they exist, are unique. The question of existence was answered in \cite[Proposition 9]{Georgiou1992}. We repeat their result here.

\begin{proposition}
    \label{prop:operator}
    Given any $T(s) \in \F^{\m \times \m}$, $\Tinf$ as given in Definition~\ref{def:ltwoinf} exists and is a closed operator.
\end{proposition}

We now show that the maximum and minimum gains of $\Tinf$ can be obtained through a frequency domain formula written in terms of the singular values of $T(i\omega)$. For the types of transfer function typically encountered in practice, this allows the maximum and minimum gains to be computed to arbitrary precision by, for example, gridding over $\omega\in\mathbb{R}$.
\begin{theorem}
    \label{thm:gain_T_inf}
    Let $T(s) \in\mathcal{F}_\infty^{m \times m}$, and let $\Tinf$ be the associated operator from Definition~\ref{def:ltwoinf}.  
    Then
    \begin{equation}    
    \begin{aligned}     & \maxgain{\Tinf}=\esssup \\  & \mingain{\Tinf}= \essinf .
    \end{aligned}
    \end{equation}  
\end{theorem}

\begin{remark}
Note that $\Tinf-\alpha\id=\mathbf{T}_{T(s)-\alpha{}I}$, meaning that Theorem~\ref{thm:gain_T_inf} can also be used to find the maximum and minimum gains of $\Tinf-\alpha\id$. A similar observation holds when computing gains with Theorems~\ref{thm:gain_T_tau}~and~\ref{thm:LMIs} below.
\end{remark}

\begin{proof}
    The first part of this theorem is certainly known, but since our setting is slightly different from standard texts (e.g.~\cite[Thm. 4.4]{Zhou96}), we sketch a proof for completeness. First, it is easily shown that for any $u\in\dom{\Tinf}$ with $\norm{u}=1$, $\norm{\Tinf u}\leq \esssup$, which implies that
    \begin{equation} \label{eq:firststep}  \maxgain{\Tinf}\leq{} \esssup.
    \end{equation}   
    We now recall that a matrix $A\in\mathcal{H}_\infty ^{p\times q}$ is said to be \textit{irreducible} if the greatest common divisor of the highest order minors of $A$ equals 1, and \textit{inner} if $A(s)^*A(s)=I$.
    By \cite[Proposition~7]{Georgiou1992}, since $\Tinf$ is maximal there exists an inner and irreducible
    \begin{equation}
        G=\begin{pmatrix}
        M\\N    \end{pmatrix}\in\mathcal{H}_\infty ^{2m \times m}
    \end{equation}
    such that $\mathcal{G}_{\hat{\T}_{T(s)}} =G\mathcal{H}_2^m$ and $T=NM^{-1}$. A standard construction shows that for any $v\in\mathbb{C}^{m}$ and almost all $\omega$, there exists a sequence $\hat{w}_k\in\mathcal{H}_2^m$ such that $\norm{\hat{w}_k}\neq0$ and
    \begin{equation}   \lim_{k\rightarrow{}\infty}\frac{\norm{Nw_k}}{\norm{Mw_k}}=\frac{\norm{N(i\omega)v}}{\norm{M(i\omega)v}}.
    \end{equation}
    It then follows that for almost all $\omega$, $\maxgain{T(i\omega)}\leq\maxgain{\Tinf}$, as required. To prove the second claim, we show that
    \begin{equation}
    \mingain{\Tinf}=\begin{cases}
    0 \quad\quad {\text{if $T(s)$ is not invertible in $\mathcal{F}_\infty^{m\times m}$};}\\
    1/\maxgain{\Tinfinverse} \quad\quad\quad\quad\quad\quad\:\: \text{otherwise}.
    \end{cases}
    \end{equation}
    This is sufficient since in both cases, the right-hand-side of the above is equal to $\essinf$ (the second case follows from the formula we have already derived for $\maxgain{\Tinf}$). Suppose first that $T$ is not invertible. 
This also implies that $N$ is not invertible, and therefore there exists a non-zero $\hat{w}\in\mathcal{H}_2^m$ such that $N\hat{w}=0$. Since $G$ is inner, $\norm{G\hat{w}}=\norm{\hat{w}}\implies{}M\hat{w}\neq{}0$. Hence
    \begin{equation}
    \begin{pmatrix}
        M\hat{w}\\0     \end{pmatrix}=G\hat{w}\in\mathcal{G}_{\hat{\T}}.
    \end{equation}
     This establishes that if $T$ is not invertible, $\mingain{\Tinf}=0$. Suppose now that $T(s)$ is invertible and consider the subspace
    \begin{equation}    
    \begin{pmatrix}
    N\\M
    \end{pmatrix}\mathcal{H}_2^m.
    \end{equation}
    Since the symbol that defines this subspace is irreducible and inner, it then follows from Proposition~\ref{prop:operator} and \cite[Proposition~7]{Georgiou1992} that this subspace is the graph of $\Tinfinverse$.
    We then see that
\begin{align}
\label{mingain_by_inverse_inf}
\mingain{\Tinf} & =\inf_{\substack{\hat{u}\in\dom{\hat{T}} \\ \hat{u}\neq0}} \frac{\norm{T\hat{u}}}{\norm{\hat{u}}}
= \inf_{\substack{\hat{w}\in\mathcal{H}_2^m \\ \hat{w}\neq 0}} \frac{\norm{N\hat{w}}}{\norm{M\hat{w}}}=1/\left(\sup_{\substack{\hat{w}\in\mathcal{H}_2^m \\ \hat{w}\neq 0}}\frac{\norm{M\hat{w}}}{\norm{N\hat{w}}}\right)=1/\maxgain{\Tinfinverse},
\end{align}
as required.
\end{proof}

\subsection{\gls{lti} operators acting on truncations of $\ltwoinf$}\label{sec:ltwotau}

Operators on extended spaces play a central role in the stability theory of dynamical systems. In this setting, the output properties of systems are characterised over all inputs in $\ltwotau$, for every $\tau>0$. Stability notions then correspond to checking whether the gain between these inputs and outputs remains bounded as $\tau\rightarrow{}\infty$. To allow \gls{srg} based tools to be applied in this setting, we consider the truncation of the operators from the previous section to $\ltwotau$.

\begin{definition}
    \label{def:ltwotau}
    Let $T\in\F^{\m\times{}\m}$, and suppose $\Tinf$ is given by Definition~\ref{def:ltwoinf}. For every $\tau>0$, define
    \[  \Ttau\coloneqq{}\trunc{}\Tinf{}\trunc.
    \]
\end{definition}

There are two main obstacles preventing the application of our results to the operators $\Ttau$.
\begin{enumerate}
    \item The operators $\Ttau$ as given by Definition~\ref{def:ltwotau} are not necessarily closed operators.
    \item The outputs of $\Ttau$ are not necessarily defined for all inputs in $\ltwotau$.
\end{enumerate}
As we will now explain, these issues are resolved through the notion of \textit{causal extendibility}. An operator $\T:\dom{\T}\subseteq{}\ltwoinf\rightarrow{}\ltwoinf$ is said to be \textit{causal} if
\[
\trunc{}\T{}u=\trunc{}\T{}\trunc{}u\;\text{for all}\,u\in\dom{\T}\,\text{and}\,\tau>0,
\]
and \textit{causally extendible} if it is causal and
\[
    \trunc \dom{\T} = \ltwotau \text{ for all } \tau > 0.
\]
As explained in \cite[\S{}7]{Georgiou1992}, the causally extendible operators are precisely the operators that can be extended uniquely to an operator on the \textit{extended $\ltwoinf$} space (the set of functions $f(t)$ defined on $[0,\infty)$ such that $\norm{\trunc f(t) } < \infty$ for all $\tau > 0$). The following result now shows this concept precisely resolves the two issues in applying our tools to the operators from Definition~\ref{def:ltwotau}.

\begin{proposition}
    Let $T\in\F^{\m\times{}\m}$, and suppose that $\Tinf$ and $\Ttau$ are given by Definitions~\ref{def:ltwoinf}~and~\ref{def:ltwotau} respectively. Then $\Ttau$ is a closed operator with $\dom{\Ttau}=\ltwotau$ for every $\tau>0$ if and only if $\Tinf$ is causally extendible.
\end{proposition}
\begin{proof}
Since $\Ttau$ is closed with $\dom{\Ttau}=\ltwotau$ and we can equivalently define $\Ttau$ as an operator from $\ltwotau$ to $\ltwotau$ (which is itself a Hilbert space), by the closed graph theorem $\Ttau:\ltwotau\rightarrow{}\ltwotau$ is bounded for every $\tau>0$. Hence $\Tinf$ is causally extendible by \cite[Proposition 12]{Georgiou1992}. Conversely if $\Tinf$ is causally extendible, $\dom{\Ttau}=\ltwotau$ by definition, and closed by \cite[Proposition 12]{Georgiou1992}. 
\end{proof}

\begin{remark}
Many frequently encountered operators are causally extendible. For example if $T(s)$ is proper and real-rational, or in $\Hf^{\m\times{}\m}$, then $\Tinf$ is causally extendible. Necessary and sufficient conditions for the causal extendibility of elements in $\F^{\m\times{}\m}$ are given in \cite[Proposition 11]{Georgiou1992}.
\end{remark}

We now turn our attention to computing the \glspl{srg} of the operators $\Ttau$. First observe that given any causal operator $\Tinf$, for any $\tau_2\geq{}\tau_1>0$,
\begin{equation}\label{eq:graphorder}
\mathbf{P}_{\tau_1}\graph{\Tinf}\subseteq\mathbf{P}_{\tau_2}\graph{\Tinf}.
\end{equation}
This shows that for any $0 < \tau_1 \leq \tau_2 \leq \tau_3 \leq \ldots$, 
\[
\T_{T(s)}^{\tau_1} \preceq \T_{T(s)}^{\tau_2} \preceq \T_{T(s)}^{\tau_3} \preceq \ldots
\]
is a monotone sequence of operators. Hence Algorithm~\ref{alg:outer} can be used to compute the \gls{srg} when $\tau \rightarrow \infty$ according to Theorem~\ref{circle_thm}. Equation~\eqref{eq:graphorder} also shows that 
\begin{equation}
    \label{eq:set_inclusion}
   \srg{\Tinf^{\tau_1}} \subseteq \srg{\Tinf^{\tau_2}} \text{ for all } \tau_1 \leq \tau_2. 
\end{equation}

We now show that the maximum and minimum gains in \eqref{eq:neededgain} for this type of operator can be calculated using frequency domain formulas based on the singular values of $T(s)$.

\begin{theorem}
    \label{thm:gain_T_tau}
    Let $T(s) \in\mathcal{F}_\infty^{m \times m}$, and $\Tinf$ be the associated operator from Definition~\ref{def:ltwoinf}. Now suppose that $\Tinf$ is causally extendible and let $\Ttau$ be given by Definition~\ref{def:ltwotau}.
    Then
    \begin{equation}
    \begin{aligned}
        & \lim_{\tau \rightarrow \infty} \maxgain{\Ttau}=\sup_{\Re s > 0}\maxgain{T(s)} \\
        & \lim_{\tau \rightarrow \infty} \mingain{\Ttau}=\inf_{\Re s > 0}\mingain{T(s)}.
    \end{aligned}
    \end{equation}  

\end{theorem}

\begin{proof}
The first result is well known, and so we only focus on the second.
We start by finding an extension of the operator $\Tinf$ that acts on inputs that may be exponentially growing. Recall from the proof of Theorem~\ref{thm:gain_T_inf} that 
\begin{equation}
    \mathcal{G}_{\Tinf} = G\htwo \text{ for some } G = \begin{pmatrix}
    M\\N
\end{pmatrix} \in \mathcal{H}_\infty^{2m \times m},
\end{equation}
such that $G$ is inner and irreducible. 
Now let $\mathcal{H}_2^\sigma=\{\hat{w}(s):\hat{w}(s+\sigma{})\in\mathcal{H}_2\}$, and consider the operator $\Tinfshat$ with graph $\graph{\Tinfshat}=G\htwo^\sigma.$
Since for any $\sigma\geq{}0$, $\htwo\subseteq{}\htwo^\sigma$, it follows that
\begin{equation}\label{eq:incpro}
\graph{\Tinfhat}\subseteq{}\graph{\Tinfshat}\;\;\text{and}\;\;
\begin{pmatrix}
    \hat{u}\\\hat{y}
\end{pmatrix}\in\graph{\Tinfhat}\implies
\begin{pmatrix}
    \hat{u}\\\hat{y}
\end{pmatrix}\in\graph{\Tinfshat}.
\end{equation}
That is, for inputs in the domain of $\Tinfhat$ the output of both operators agree, but the domain of $\Tinfshat$ includes some additional inputs in $\htwo^\sigma$. We may similarly extend $\Tinf$. Let $\mathcal{L}_2^\sigma=\{f:fe^{-\sigma{}t}\in\mathcal{L}_2\}$, and $\mathbf{F}:\ltwoinf\rightarrow{}\htwo^{\m}$ denote the Fourier transform. Define $\Tinfs$ to be the operator with graph equal to the mapping of the graph of $\Tinfshat$ through the transformations
\begin{center}
\begin{tikzcd}
\mathcal{L}_2^\sigma \arrow[r, shift left, "e^{-\sigma{}t}f"] & 
\mathcal{L}_2 \arrow[r, shift left, "\mathbf{F}"] \arrow[l, shift left, "e^{\sigma{}t}f", {yshift=-2pt}]&
\htwo{} \arrow[r, shift left, "\hat{f}(s-\sigma)"] \arrow[l, shift left, "\mathbf{F}^{-1}", {yshift=-2pt}]&
\htwo^\sigma. \arrow[l, shift left, "\hat{f}(s+\sigma)" {yshift=-2pt}]
\end{tikzcd}
\end{center}
It is easily shown that $\Tinfs$ extends $\Tinf$. However, since $\trunc{}f\in\ltwotau$ if and only if $\trunc{}e^{\sigma{}t}f\in\ltwotau$, it also follows that
\begin{equation}\label{eq:keytrunc}
\trunc{}\graph{\Tinfs}=\trunc{}\graph{\Tinf}.
\end{equation}
We will now use this fact to demonstrate that particular exponentially growing functions belong to $\trunc{}\graph{\Tinf}$. Let $s'\in\Comp$ and $v\in\Comp^m$, and suppose that $\mathrm{Re}(s')>0$. For any $\sigma>\mathrm{Re}(s')$,
\[
\hat{w}=v\frac{1}{s-s'}\in\htwo^\sigma.
\]
Hence $G\hat{w}\in\graph{\Tinfshat}$. Putting
\[
\hat{z}=\left(G(s)-G(s')\right)v\frac{1}{s-s'},
\]
we also see that
\[
G\hat{w}=Gv\frac{1}{s-s'}=
    G(s')
v\frac{1}{s-s'}+\hat{z}.
\]
Now $\hat{z}$ is analytic on the open right-half plane (note that the simple pole at $s'$ will be cancelled by a zero), and
\[
\norm{\hat{z}(i\omega)}\leq{}\frac{\norm{v}(\norm{G}_{\infty}+\vert{}G(s')\vert{})}{\sqrt{(\omega-\mathrm{Im}(s'))^2+\mathrm{Re}(s')^2}},
\]
from which it can be shown that $\hat{z}\in\htwo$. Hence
\[
\begin{pmatrix}
    u\\y
\end{pmatrix}=G(s')ve^{s't}+z\in\graph{\Tinfs},
\]
where the size of $z$ becomes arbitrarily small compared to the size of $e^{{s't}}$ as $t\rightarrow{}\infty$. It then follows that
\[
\lim_{\tau\rightarrow{}\infty}\frac{\norm{\trunc{}y}}{\norm{\trunc{}u}}=\frac{\norm{N(s')v}}{\norm{M(s')v}},
\]
and since both $s'$ and $v$ were arbitrary, we conclude from \eqref{eq:keytrunc} that
\[
    \lim_{\tau\rightarrow\infty{}}\mingain{\Ttau}
    \leq{}\inf_{\text{Re}(s)>0}\mingain{T(s)}.
\]

This proves the result in the case that the infimum on the right-hand-side of the above is zero, and so we now assume that
\[
\inf_{\Re(s)>0}\mingain{T(s)}>0.
\]
This implies that $T^{-1}$ exists, and since $T=NM^{-1}$, $N^{-1}$ exists. Since $\Tinf$ is causally extendible, by~\cite[Proposition~11]{Georgiou1992} there exists an $\alpha > 0$ such that
\[
\inf_{\Re(s)\geq{}\alpha}\mingain{M(s)}>0.
\]
We know that 
\begin{equation}
    \mingain{T(s)} \leq \mingain{N(s)}\maxgain{M^{-1}(s)} = \frac{\mingain{N(s)}}{\mingain{M(s)}},
\end{equation}
which gives us $\mingain{N(s)} \geq \mingain{T(s)}\mingain{M(s)}$. This shows that
\[
\inf_{\Re(s)\geq{}\alpha}\mingain{N(s)}>0,
\]
and therefore $\Tinfinverse$ is causally extendible~\cite[Proposition~11]{Georgiou1992}. 
As $T(s)$ is invertible the graph of $\hat{\T}_{T(s)^{-1}}$ is
\begin{equation}
        \mathcal{G}_{\hat{\T}_{T(s)^{-1}}} = \begin{pmatrix}
    N\\M
\end{pmatrix}\mathcal{H}_2 
\end{equation}
and because $\Tinfinverse$ is causally extendible $\mathcal{G}_{\Ttauinverse}=\trunc\mathcal{G}_{\Tinfinverse}$. We can now follow the same argument as at the end of the proof of Theorem~\ref{thm:gain_T_inf} to conclude the proof. 
\end{proof}

\begin{remark}
    The differences between the \gls{srg} of $\Tinf$ and the \gls{srg} of $\Ttau$ when $\tau \rightarrow \infty$ are closely related to the soft and hard \gls{srg} as defined in~\cite{chen2025a}. However, in~\cite{chen2025a} they define the operators over the extended $\ltwoinf$ space, so the definitions are not equivalent. 
    The $\mathrm{SRG}(\Tinf)$ is the same set as the soft \gls{srg}, as the soft \gls{srg} is, similarly to $\Tinf$, defined over the subset of the operator for which the output is in $\ltwoinf$. From~\eqref{eq:set_inclusion} it follows that the set $\lim_{\tau \rightarrow \infty} \mathrm{SRG}(\Ttau)$ is almost equal to the hard \gls{srg}, but as the set limit is closed by definition $\lim_{\tau \rightarrow \infty} \mathrm{SRG}(\Ttau)$ is the closure of the hard \gls{srg}, which is itself not necessarily closed.
\end{remark}

\subsection{Operators with state-space representations}\label{sec:ss}
A subclass of the transfer matrices $T(s) \in \mathcal{F}_\infty^{m \times m}$ is the proper, real-rational transfer matrices, which can be written on state-space form $(A,B,C,D)$ such that \[T(s) = \C(sI-\A)^{-1}\B + \D.\] 
In the time domain the outputs of the operators $\Tinf$ and $\Ttau$ are defined through the equation
\begin{equation}
    \label{ss_time}
    \y(t) = \int_0^t \C e^{\A(t-t_1)}\B u(t_1)dt_1 + \D u(t).
\end{equation}

When the transfer functions are specified by their state-space representation, the maximum gain for the corresponding operators can be computed using \glspl{lmi} following the Bounded Real Lemma. The minimum gain can also be computed using this result, as explained in Remark~\ref{rem:mingain_inf}, which is stated after the Bounded Real Lemma.
\begin{theorem}[Bounded Real Lemma]
    Consider $T(s) = \C(sI-\A)^{-1}\B + \D \in \mathcal{F}_\infty^{m \times m}$
    such that $(\A,\B)$ is controllable and $(\A,\C)$ is observable, and let $\Tinf$ and $\Ttau$ be given as in Definitions~\ref{def:ltwoinf} and~\ref{def:ltwotau}. If $\A$ has imaginary axis eigenvalues, then $\maxgain{\Tinf}\!=\!\lim_{\tau \rightarrow \infty}\maxgain{\Ttau}\!=\!\infty$, else 
    \begin{align}    
    & \lim_{\tau \rightarrow \infty} \maxgain{\Ttau} = \inf \gamma \\
    & \text{such that there exists a}\;\; \P = \P^\top \succeq 0 \;\; \text{satisfying}: \\
    & \begin{bmatrix}
        \A^\top \P + \P \A + \C^\top \C & \P \B +  \C^\top \D   \\
        \B^\top \P + \D^\top \C & \D^\top \D - \gamma^2 I
    \end{bmatrix}  \preceq 0.
    \label{LMI upper bound}
    \end{align}  
    If we remove the constraint $\P \succeq 0$ from~\eqref{LMI upper bound}, the \gls{lmi} gives $\maxgain{\Tinf} = \inf \gamma$ instead. 
\label{thm:LMIs}
\end{theorem}

\begin{remark}
    \label{rem:mingain_inf}
    It follows from the proofs of Theorems~\ref{thm:gain_T_inf}~and~\ref{thm:gain_T_tau}, that in the state-space case, the minimum gain of the operators $\Tinf$ and $\Ttau$ can also be computed using Theorem~\ref{thm:LMIs}. In particular, if $D$ is not invertible, the minimum gain of both operators is equal to zero. Otherwise, they are equal to $1/\inf{}\gamma$, where $\inf{}\gamma$ is computed as in Theorem~\ref{thm:LMIs}, and
    \[
    A\leftarrow{}A-BD^{-1}C,\,B\leftarrow{}BD^{-1},\,C\leftarrow{}-D^{-1}C,D\leftarrow{}D^{-1}
    \]
    have been substituted into the corresponding \gls{lmi}.  
    Alternatively, we can obtain the minimum gain by solving a related \gls{lmi} following~\cite{Bridgeman} (see also \cite[\S{4}]{degroot2025}). 
\end{remark}
\section{Examples}
\label{sec:examples}
In this section, we illustrate examples of the closures of the \gls{srg} for a variety of operators that are defined by $T(s) \in \mathcal{F}_\infty^{m \times m}$, and discuss the difference between the \glspl{srg} of $\Tinf$ and $\Ttau$ when $\tau \rightarrow \infty$. We also present how \glspl{srg} can be used for stability analysis of dynamical systems.

\subsection{Examples of \glspl{srg} for dynamical systems}
The specific operators we will consider are
\begin{equation}
    T_1(s) = \frac{e^{-s}}{s+1}, \, T_2(s) = \begin{bmatrix}
         \frac{1}{s-1} & \frac{-2}{s+3} \\ \frac{1}{s-1} & \frac{1}{s+2}
     \end{bmatrix} \text{ and } T_3(s) = \frac{1}{s}. 
\end{equation}
Their \glspl{srg} for both $\Tinf$ and $\Ttau$ when $\tau \rightarrow \infty$ are shown in Figures~\ref{fig:delay},~\ref{fig:MIMO} and~\ref{fig:integrator}.
\begin{figure}
    \centering
    \begin{tikzpicture}[every to/.style={hyperbolic plane}, scale=3,>={Stealth[scale=1]}]     
    \input{delay_simple_T_tau.tex}
    \input{delay_Tinf.tex}
    \input{delay_Tinf_pattern.tex}

      \foreach \x in { -0.5, 0.5, 1}{
        \node[below] at (\x,-0.1) {\x};
        \draw (\x,0) -- (\x,-0.05);
      }

    \foreach \y in {-0.5, 0.5}{
        \node[left] at (-0.05,\y) {\y};
        \draw (0,\y) -- (-0.05,\y);
    }


      \draw[->] (-0.8,0) -- (1.5,0) node[right] {$\Re$};
      \draw[->] (0,-1) -- (0,1) node[above] {$\Im$};
    \end{tikzpicture}
    \caption{The figure shows the \gls{srg} of $T_1(s) = \frac{e^{-s}}{s+1}$. The hatched grey area shows the \gls{srg} of $\Tinf$ and the orange area shows the \gls{srg} of $\Ttau$ when $\tau \rightarrow \infty$.}
    \label{fig:delay}
\end{figure}
\begin{figure}
     \centering
    \begin{tikzpicture}[scale=2.2,>={Stealth[scale=1]}]
    \fill[DutchOrange!70] (-1.65,-1.35) rectangle (1.45,1.33);

      \draw[fill=white]  (-1.3599, 0.0000) to[hyperbolic plane] (-1.3581, 0.0031) to[hyperbolic plane] (-1.3574, 0.0081) to[hyperbolic plane] (-1.3566, 0.0128) to[hyperbolic plane] (-1.3559, 0.0176) to[hyperbolic plane] (-1.3550, 0.0226) to[hyperbolic plane] (-1.3542, 0.0276) to[hyperbolic plane] (-1.3532, 0.0328) to[hyperbolic plane] (-1.3523, 0.0381) to[hyperbolic plane] (-1.3512, 0.0435) to[hyperbolic plane] (-1.3502, 0.0492) to[hyperbolic plane] (-1.3490, 0.0549) to[hyperbolic plane] (-1.3478, 0.0609) to[hyperbolic plane] (-1.3465, 0.0670) to[hyperbolic plane] (-1.3451, 0.0733) to[hyperbolic plane] (-1.3441, 0.0781) to[hyperbolic plane] (-1.3430, 0.0829) to[hyperbolic plane] (-1.3414, 0.0897) to[hyperbolic plane] (-1.3397, 0.0968) to[hyperbolic plane] (-1.3380, 0.1040) to[hyperbolic plane] (-1.3361, 0.1115) to[hyperbolic plane] (-1.3340, 0.1193) to[hyperbolic plane] (-1.3319, 0.1274) to[hyperbolic plane] (-1.3296, 0.1357) to[hyperbolic plane] (-1.3271, 0.1444) to[hyperbolic plane] (-1.3245, 0.1533) to[hyperbolic plane] (-0.6803, 0.5320) to[hyperbolic plane] (-0.4063, 0.3889) to[hyperbolic plane] (-0.3897, 0.3728) to[hyperbolic plane] (-0.3760, 0.3590) to[hyperbolic plane] (-0.3641, 0.3468) to[hyperbolic plane] (-0.3536, 0.3358) to[hyperbolic plane] (-0.3441, 0.3256) to[hyperbolic plane] (-0.3353, 0.3161) to[hyperbolic plane] (-0.3271, 0.3071) to[hyperbolic plane] (-0.3194, 0.2987) to[hyperbolic plane] (-0.3121, 0.2906) to[hyperbolic plane] (-0.3052, 0.2829) to[hyperbolic plane] (-0.2985, 0.2755) to[hyperbolic plane] (-0.2920, 0.2683) to[hyperbolic plane] (-0.2858, 0.2613) to[hyperbolic plane] (-0.2797, 0.2546) to[hyperbolic plane] (-0.2738, 0.2480) to[hyperbolic plane] (-0.2680, 0.2416) to[hyperbolic plane] (-0.2623, 0.2354) to[hyperbolic plane] (-0.2568, 0.2292) to[hyperbolic plane] (-0.2513, 0.2232) to[hyperbolic plane] (-0.2459, 0.2173) to[hyperbolic plane] (-0.2405, 0.2115) to[hyperbolic plane] (-0.2352, 0.2058) to[hyperbolic plane] (-0.2299, 0.2001) to[hyperbolic plane] (-0.2247, 0.1945) to[hyperbolic plane] (-0.2195, 0.1890) to[hyperbolic plane] (-0.2143, 0.1836) to[hyperbolic plane] (-0.2092, 0.1782) to[hyperbolic plane] (-0.2040, 0.1729) to[hyperbolic plane] (-0.1989, 0.1676) to[hyperbolic plane] (-0.1937, 0.1623) to[hyperbolic plane] (-0.1886, 0.1571) to[hyperbolic plane] (-0.1834, 0.1519) to[hyperbolic plane] (-0.1782, 0.1468) to[hyperbolic plane] (-0.1730, 0.1417) to[hyperbolic plane] (-0.1678, 0.1366) to[hyperbolic plane] (-0.1626, 0.1315) to[hyperbolic plane] (-0.1573, 0.1265) to[hyperbolic plane] (-0.1520, 0.1215) to[hyperbolic plane] (-0.1466, 0.1165) to[hyperbolic plane] (-0.1412, 0.1115) to[hyperbolic plane] (-0.1358, 0.1065) to[hyperbolic plane] (-0.1303, 0.1015) to[hyperbolic plane] (-0.1248, 0.0966) to[hyperbolic plane] (-0.1192, 0.0916) to[hyperbolic plane] (-0.1136, 0.0867) to[hyperbolic plane] (-0.1078, 0.0818) to[hyperbolic plane] (-0.1021, 0.0768) to[hyperbolic plane] (-0.0962, 0.0719) to[hyperbolic plane] (-0.0903, 0.0670) to[hyperbolic plane] (-0.0843, 0.0621) to[hyperbolic plane] (-0.0783, 0.0572) to[hyperbolic plane] (-0.0721, 0.0523) to[hyperbolic plane] (-0.0659, 0.0474) to[hyperbolic plane] (-0.0596, 0.0425) to[hyperbolic plane] (-0.0532, 0.0376) to[hyperbolic plane] (-0.0467, 0.0326) to[hyperbolic plane] (-0.0401, 0.0277) to[hyperbolic plane] (-0.0334, 0.0228) to[hyperbolic plane] (-0.0266, 0.0178) to[hyperbolic plane] (-0.0197, 0.0127) to[hyperbolic plane] (-0.0126, 0.0074) to[hyperbolic plane] (-0.0059, 0.0000);

\begin{scope}[yscale=-1]
    \draw[fill=white]  (-1.3599, 0.0000) to[hyperbolic plane] (-1.3581, 0.0031) to[hyperbolic plane] (-1.3574, 0.0081) to[hyperbolic plane] (-1.3566, 0.0128) to[hyperbolic plane] (-1.3559, 0.0176) to[hyperbolic plane] (-1.3550, 0.0226) to[hyperbolic plane] (-1.3542, 0.0276) to[hyperbolic plane] (-1.3532, 0.0328) to[hyperbolic plane] (-1.3523, 0.0381) to[hyperbolic plane] (-1.3512, 0.0435) to[hyperbolic plane] (-1.3502, 0.0492) to[hyperbolic plane] (-1.3490, 0.0549) to[hyperbolic plane] (-1.3478, 0.0609) to[hyperbolic plane] (-1.3465, 0.0670) to[hyperbolic plane] (-1.3451, 0.0733) to[hyperbolic plane] (-1.3441, 0.0781) to[hyperbolic plane] (-1.3430, 0.0829) to[hyperbolic plane] (-1.3414, 0.0897) to[hyperbolic plane] (-1.3397, 0.0968) to[hyperbolic plane] (-1.3380, 0.1040) to[hyperbolic plane] (-1.3361, 0.1115) to[hyperbolic plane] (-1.3340, 0.1193) to[hyperbolic plane] (-1.3319, 0.1274) to[hyperbolic plane] (-1.3296, 0.1357) to[hyperbolic plane] (-1.3271, 0.1444) to[hyperbolic plane] (-1.3245, 0.1533) to[hyperbolic plane] (-0.6803, 0.5320) to[hyperbolic plane] (-0.4063, 0.3889) to[hyperbolic plane] (-0.3897, 0.3728) to[hyperbolic plane] (-0.3760, 0.3590) to[hyperbolic plane] (-0.3641, 0.3468) to[hyperbolic plane] (-0.3536, 0.3358) to[hyperbolic plane] (-0.3441, 0.3256) to[hyperbolic plane] (-0.3353, 0.3161) to[hyperbolic plane] (-0.3271, 0.3071) to[hyperbolic plane] (-0.3194, 0.2987) to[hyperbolic plane] (-0.3121, 0.2906) to[hyperbolic plane] (-0.3052, 0.2829) to[hyperbolic plane] (-0.2985, 0.2755) to[hyperbolic plane] (-0.2920, 0.2683) to[hyperbolic plane] (-0.2858, 0.2613) to[hyperbolic plane] (-0.2797, 0.2546) to[hyperbolic plane] (-0.2738, 0.2480) to[hyperbolic plane] (-0.2680, 0.2416) to[hyperbolic plane] (-0.2623, 0.2354) to[hyperbolic plane] (-0.2568, 0.2292) to[hyperbolic plane] (-0.2513, 0.2232) to[hyperbolic plane] (-0.2459, 0.2173) to[hyperbolic plane] (-0.2405, 0.2115) to[hyperbolic plane] (-0.2352, 0.2058) to[hyperbolic plane] (-0.2299, 0.2001) to[hyperbolic plane] (-0.2247, 0.1945) to[hyperbolic plane] (-0.2195, 0.1890) to[hyperbolic plane] (-0.2143, 0.1836) to[hyperbolic plane] (-0.2092, 0.1782) to[hyperbolic plane] (-0.2040, 0.1729) to[hyperbolic plane] (-0.1989, 0.1676) to[hyperbolic plane] (-0.1937, 0.1623) to[hyperbolic plane] (-0.1886, 0.1571) to[hyperbolic plane] (-0.1834, 0.1519) to[hyperbolic plane] (-0.1782, 0.1468) to[hyperbolic plane] (-0.1730, 0.1417) to[hyperbolic plane] (-0.1678, 0.1366) to[hyperbolic plane] (-0.1626, 0.1315) to[hyperbolic plane] (-0.1573, 0.1265) to[hyperbolic plane] (-0.1520, 0.1215) to[hyperbolic plane] (-0.1466, 0.1165) to[hyperbolic plane] (-0.1412, 0.1115) to[hyperbolic plane] (-0.1358, 0.1065) to[hyperbolic plane] (-0.1303, 0.1015) to[hyperbolic plane] (-0.1248, 0.0966) to[hyperbolic plane] (-0.1192, 0.0916) to[hyperbolic plane] (-0.1136, 0.0867) to[hyperbolic plane] (-0.1078, 0.0818) to[hyperbolic plane] (-0.1021, 0.0768) to[hyperbolic plane] (-0.0962, 0.0719) to[hyperbolic plane] (-0.0903, 0.0670) to[hyperbolic plane] (-0.0843, 0.0621) to[hyperbolic plane] (-0.0783, 0.0572) to[hyperbolic plane] (-0.0721, 0.0523) to[hyperbolic plane] (-0.0659, 0.0474) to[hyperbolic plane] (-0.0596, 0.0425) to[hyperbolic plane] (-0.0532, 0.0376) to[hyperbolic plane] (-0.0467, 0.0326) to[hyperbolic plane] (-0.0401, 0.0277) to[hyperbolic plane] (-0.0334, 0.0228) to[hyperbolic plane] (-0.0266, 0.0178) to[hyperbolic plane] (-0.0197, 0.0127) to[hyperbolic plane] (-0.0126, 0.0074) to[hyperbolic plane] (-0.0059, 0.0000); 
\end{scope}

      \begin{scope}[every to/.style={hyperbolic plane}]
         \input{EJC2.tex}
         \input{EJC2_pattern.tex}
      \end{scope}

      \foreach \x in {-1.5, -1, -0.5, 0.5, 1}{
        \node[below] at (\x,-0.02) {\x};
        \draw (\x,0) -- (\x,-0.05);
      }

    \foreach \y in {-1, -0.5, 0.5, 1}{
         \node[left] at (-0.05,\y) {\y};
         \draw (0,\y) -- (-0.05,\y);
     }


      \draw[->] (-1.6,0) -- (1.5,0) node[right] {$\Re$};
      \draw[->] (0,-1.35) -- (0,1.4) node[above] {$\Im$};
      \node[left] at (1.4,1.1) {\large{$\cup \{\infty\}$}};

    \end{tikzpicture}
    \caption{The figure shows the \gls{srg} of a multi-input multi-output system $T_2(s) = [
         \frac{1}{s-1}, \frac{-2}{s+3} ; \frac{1}{s-1}, \frac{1}{s+2} ]$. The \gls{srg} of $\Tinf$ is the hatched grey area, while the \gls{srg} of $\Ttau$ when $\tau \rightarrow \infty$ is represented by the orange area.}
    \label{fig:MIMO}
  \end{figure}
  \begin{figure}
  \centering
    \begin{tikzpicture}[>={Stealth[scale=1]}]
        \draw[ultra thick, DutchOrange!70] (0,-1.5) -- (0,1.5);
        \draw[->] (0,-1.5) -- (0,1.6) node[above] {$\Im$};

    \fill[DutchOrange!70] (0,-1.5) -- (3,-1.5) -- (3,1.5) -- (0,1.5) -- cycle;
    \fill[
      pattern={Lines[angle=45, distance=5pt, line width=0.5pt]},
      pattern color=black
    ] (-0.03,-1.5) rectangle (0.03,1.5);

    \draw[->] (-3,0) -- (3.2,0) node[right] {$\Re$};
    \node[left] at (3,1.1) {\large{$\cup \{\infty\}$}};
    
    \end{tikzpicture}
    \caption{The figure shows the \gls{srg} of $T_3(s) = 1/s$. For $\Tinf$ the \gls{srg} is the extended imaginary axis (hatched grey), while for $\Ttau$ when $\tau \rightarrow \infty$ it is the extended closed right-half-plane including the imaginary axis (orange).}
     \label{fig:integrator}
\end{figure}

The first example (Figure~\ref{fig:delay}) shows that the \glspl{srg} for the two different operators $\Tinf$ and $\Ttau$ can be quite similar if $T(s) \in \Hf^{m \times m}$. In the next example (Figure~\ref{fig:MIMO}) we can see that this is not the case when $T(s) \notin \Hf^{m \times m}$. Here the \gls{srg} of $\Ttau$ includes $\{\infty\}$, while it remains bounded for $\Tinf$. The last example (Figure~\ref{fig:integrator}) shows that the \gls{srg} of the integrator is the extended imaginary axis for $\Tinf$ and the extended closed right-half-plane for $\Ttau$. This illustrates that the \glspl{srg} of both operators include $\{\infty\}$ if $T(s)$ has imaginary axis poles. In all examples we see that in the construction of the \gls{srg} of $\Ttau$, right-half-plane zeros have the effect of filling in the interior of the \gls{srg} of $\Tinf$.

\subsection{Stability analysis using \glspl{srg}}
Now that we have shown how to compute \glspl{srg} for dynamical systems, we give a brief introduction to how these \glspl{srg} can be used for stability analysis. Following the definition of~\cite{Georgiou1992}, a dynamical system associated with the operator $\Tinf$ is said to be stable if $\dom{\Tinf} = \ltwoinf$ and $\bar{\sigma}(\Tinf)<\infty$. This is equivalent to saying that a dynamical system is stable if $\lim_{\tau \rightarrow \infty} \bar{\sigma}(\Ttau) < \infty$ for the associated operator $\Ttau$. It then follows that a dynamical system is stable if and only if
\begin{equation}
        \{\infty\} \notin \lim_{\tau \rightarrow \infty} \srg{\Ttau}.
        \label{eq:stability}
\end{equation}
Combining this result with computational rules for \glspl{srg} in~\cite{Ryu2021} allows for stability analysis of interconnected systems. 
To illustrate this we show a simple example. Consider the feedback interconnection in Figure~\ref{fig:feedback} and assume that it is well-posed. 
\begin{figure}[h]
    \centering
    \begin{tikzpicture}[>=Stealth, xscale=1, yscale=1]
        \node[draw, circle, inner sep=1.8pt] (sum) at (0,0) {$+$};
        \node[draw, rectangle, minimum width=1.5cm, minimum height=1cm] (H1) at (2.5,0) {$G(s)$};
        \node[draw, rectangle, minimum width=1.5cm, minimum height=1cm] (H2) at (2.5,-1.7) {$-I$};

        \draw[->] (-1.5,0) -- (sum) node[midway, above] {$r$};
        \draw[->] (sum) -- (H1) node[midway, above] {$e$};

        \draw[->] (H1.east) -- ++(2.5,0) node[midway, above] {$y$};

        \draw (5,0) -- ++(0,-1.7);
        \draw[->] (5,-1.7) -- (H2.east);
        \draw[->] (H2.west) -- (0,-1.7) -- (sum.south);

    \end{tikzpicture}
    \caption{Feedback interconnection of a system with transfer function $G(s)$ in feedback with $-I$.}
    \label{fig:feedback}
\end{figure}
In this interconnection the operator from $r$ to $e$ is given by $\T^\tau = (I + \T_{G(s)}^\tau)^{-1}$ for all $\tau > 0$. The following rules from~\cite[Theorems 4 and 5]{Ryu2021}
 \begin{align}
        & \srg{\mathbf{I} + \T} = 1 + \srg{\T} \\
        & \srg{\T^{-1}} = \srg{\T}^{-1}
\end{align}
then give that 
\[
\lim_{\tau \rightarrow \infty} \srg{\T^\tau} = \left(1 + \lim_{\tau \rightarrow \infty} \srg{\T_{G(s)}^\tau} \right)^{-1},
\]
as inversion is a continuous mapping under the chordal metric. 
It then follows from \eqref{eq:stability} that the feedback interconnection is stable if and only if
\begin{equation}
        \{-1\} \notin \lim_{\tau \rightarrow \infty} \srg{\T_{G(s)}^\tau}.
\end{equation}
Recall that $\lim_{\tau \rightarrow \infty} \mathrm{SRG}(\T_{G(s)}^\tau)$ is a closed set by the definition of the set limit, so we get strict separation between $\{-1\}$ and $\lim_{\tau \rightarrow \infty} \mathrm{SRG}(\T_{G(s)}^\tau)$ when the interconnection is stable. This shows that $T_1(s)$, $T_2(s)$, and $T_3(s)$ are all stable in feedback with $-I$, as we can see in Figures~\ref{fig:delay},~\ref{fig:MIMO}, and~\ref{fig:integrator} that $\{-1\}$ is not included in the \glspl{srg} of $\Ttau$ when $\tau \rightarrow \infty$.

\begin{remark}
    Note that none of the steps above require linearity and therefore the same method can be used to analyse stability of non-linear systems.
\end{remark}

\section{Conclusions}
\label{sec:conclusions}
In this paper, we have shown how to construct the \gls{srg} of a closed linear, possibly unbounded, operator using an algorithm based on maximum and minimum gain computations. To illustrate the utility of the result, we specified how to compute the \gls{srg} for two types of operators, defined over $\ltwoinf$ and truncations of $\ltwoinf$, that model dynamical systems. Finally, we provided a set of \glspl{lmi} for computing the \gls{srg} for the special case when the models are given on state-space form.

\bibliographystyle{plain}
\bibliography{main}

\end{document}